\documentclass[twocolumn,amsthm]{autart}

\newif\iffull
\fulltrue

\usepackage[utf8]{inputenc}
\usepackage[T1]{fontenc}
\usepackage{mathtools}
\usepackage{amssymb}
\usepackage{aut_thm}
\usepackage{comment}
\usepackage{des}
\usepackage{natbib}
\usepackage{url}

\setcitestyle{authoryear}

\usepackage[dvipsnames]{xcolor}
\usepackage{todonotes}

\newcommand{\G}{\ensuremath{G}\xspace}
\newcommand{\C}{\ensuremath{C}\xspace}
\newcommand{\Aq}{\ensuremath{\mathsf{A}}\xspace}
\newcommand{\Bq}{\ensuremath{\mathsf{B}}\xspace}
\newcommand{\Cq}{\ensuremath{\mathsf{C}}\xspace}
\newcommand{\Dq}{\ensuremath{\mathsf{D}}\xspace}

\newcommand{\enum}[1]{\ensuremath{E_{#1}}}
\newcommand{\enuma}{\enum{0}} 
\newcommand{\enumb}{\enum{5}} 
\newcommand{\enumc}{\enum{1}} 
\newcommand{\enumd}{\enum{9}} 
\newcommand{\enume}{\enum{8}} 
\newcommand{\enumf}{\enum{7}} 
\newcommand{\enumg}{\enum{11}} 
\newcommand{\enumh}{\enum{10}} 
\newcommand{\enumi}{\enum{6}} 
\newcommand{\enumj}{\enum{4}} 
\newcommand{\enumk}{\enum{2}} 
\newcommand{\enuml}{\enum{3}} 

\newcommand{\home}{\ensuremath{\mathsf{Home}}\xspace}
\newcommand{\task}[1]{\ensuremath{\mathsf{Task}_{#1}}\xspace}
\newcommand{\taskc}{\task{1}}
\newcommand{\taskk}{\task{2}}
\newcommand{\taskl}{\task{3}}

\usepackage[ruled,lined,linesnumbered,norelsize,algo2e]{algorithm2e}
\renewenvironment{algorithm}[1][tbp]{
  \begin{algorithm2e}[#1]\ignorespaces%
}
{
  \end{algorithm2e}\ignorespacesafterend%
}
\newenvironment{algorithm*}[1][tbp]{
  \begin{algorithm2e*}[#1]\ignorespaces%
}
{
  \end{algorithm2e*}\ignorespacesafterend%
}

\usepackage{tikz}
\usetikzlibrary{automata}
\usetikzlibrary{positioning}
\usetikzlibrary{fit}
\usetikzlibrary{shapes}
\usetikzlibrary{calc}
\usetikzlibrary{arrows.meta}
\usetikzlibrary{shapes.geometric}

\newcommand{\rdsb}{robust stutter bisimulation\xspace}

\iffull
	\newcommand\appref[1]{#1}
	\newcommand\shorttext[1]{}
\else
    \newcommand\appref[1]{}
    \newcommand\shorttext[1]{#1}
\fi

\begin{document}

\begin{frontmatter}

\iffull
    \title{Robust Stutter Bisimulation for Abstraction \\ and Controller Synthesis with Disturbance: Proofs\thanksref{wasp}}
\else
    \title{Robust Stutter Bisimulation for Abstraction \\ and Controller Synthesis with Disturbance\thanksref{wasp}}
\fi

\thanks[wasp]{This work was partially supported by the Wallenberg AI, Autonomous Systems and Software Program (WASP) funded by the Knut and Alice Wallenberg Foundation.}

\author[chalmers,zenseact]{Jonas Krook\corauthref{cor}}%
  \ead{krookj@chalmers.se}\quad
\author[waikato]{Robi Malik}%
  \ead{robi@waikato.ac.nz}\quad
\author[chalmers]{Sahar Mohajerani}%
  \ead{mohajera@chalmers.se}\quad
\author[chalmers]{Martin Fabian}%
  \ead{fabian@chalmers.se} 
  
\corauth[cor]{Corresponding author.}

\address[chalmers]{Department of Electrical Engineering,
	Chalmers University of Technology, Göteborg, Sweden}

\address[waikato]{Department of Software Engineering,
	University of Waikato, Hamilton, New Zealand}

\address[zenseact]{Zenseact, Göteborg, Sweden}


\begin{abstract}
This \whatsit\ proposes a method to synthesise controllers for cyber-physical systems such that the controlled systems satisfy specifications given as linear temporal logic formulas.
The focus is on systems with \emph{disturbance}, where future states cannot be predicted exactly due to uncertainty in the environment.
The approach used to solve this problem is to first construct a finite-state abstraction of the original system and then synthesise a controller for the abstract system.
For this approach, the \emph{\rdsb} relation is introduced, which preserves the existence of controllers for any given linear temporal logic formula. States are related by the \rdsb relation if the same target sets can be guaranteed to be reached or avoided under control of some controllers, thereby ensuring that disturbances have similar effect on paths that start in related states.
This \whatsit\ presents an algorithm to construct the corresponding \rdsb
\emph{quotient} to solve the abstraction problem, and it is shown, by
explicit construction, that there exists a controller enforcing a linear
temporal logic formula for the original system if and only if a
corresponding controller exists for the quotient system. Lastly, the result
of the algorithm and the controller construction are demonstrated by
application to an example of robot navigation.
\end{abstract}

\end{frontmatter}

\section{Introduction}

\emph{Cyber-physical systems} consist of physical systems and digital computers affecting and interacting with each other \citep{Lee:CPS:2015}. This can be a continuous time dynamical system subject to control inputs and process disturbances, which is under control of a computer with a fixed sampling time. Cyber-physical systems are often safety critical, thus it is highly desirable to have correctness guarantees.
One way to achieve such guarantees is to use \emph{formal synthesis} to construct the control logic automatically \citep{BelYorGol:17}.

Formal synthesis uses a model of the cyber-physical system and a \emph{formal specification} to compute the allowed control actions such that the behaviour of the controlled system satisfies the specification. The specification is a formalisation of the requirements, and it can be expressed in, for instance, \emph{Linear Temporal Logic} (\LTL, \citealt{BaiKat:08}). \LTL\ extends propositional logic with temporal operators with which requirements on future behaviour can be specified.

Standard algorithms \citep{Ram:89, KloBel:08, BelYorGol:17}
for controller synthesis with \LTL\ specifications require finite transition systems, whereas many cyber-physical systems are described by continuous models with an infinite state space. Such models can be turned into transition systems by \emph{discretisation}, but in general neither the state space nor the transition relation of the resulting \emph{concrete} transition system are finite~\citep{Tab:09}. One approach to reduce the size of the state space is by attempting to form a finite-state \emph{abstraction} by grouping states into a finite \emph{quotient} state space. One such abstraction method is \emph{bisimulation}~\citep{Mil:89}. It preserves all \LTL\ properties~\citep{BaiKat:08}, and is guaranteed to produce finite quotient spaces for certain types of systems~\citep{AluHenLafPap:00}.

If the bisimulation quotient of a system is infinite, a finite quotient space might be obtained by using a coarser abstraction. One approach to obtain coarser quotients is \emph{approximate} bisimulation where bisimulation is relaxed to allow a bounded difference between the behaviours of the concrete and abstract system \citep{GirPap:07}. Coarser quotients also result from (approximate) \emph{simulation} \citep{Tab:06, ZamPolMazTab:11, BelYorGol:17, ReiWebRun:16}, which relaxes bisimulation by retaining only some controlled behaviours of the concrete system to the abstract system.
\emph{Dual-simulation} \citep{WagOza:16} produces a coarser abstraction than bisimulation by using overlapping subsets.
Another coarse abstraction is obtained by \emph{divergent stutter bisimulation} which allows \emph{stutter steps} within the abstract states~\citep{BaiKat:08}. \citet{MohMalWinLafOza:21} show that divergent stutter bisimulation yields abstractions for which a controller can be synthesised if and only if a controller can be synthesised for the concrete system, given that the synthesis is performed with \LTL\ specifications without the next operator~(\LTLnn). Omitting this operator means that specifications cannot refer to specific time intervals, but this is not always necessary for formalising the requirements.
All the above work applies to \emph{deterministic} transition systems and does not work for systems that are subject to process disturbances.

Many concrete systems are subject to disturbances that either stem from the system dynamics or fidelity loss in the discretisation. Abstraction methods for controller synthesis for non-deterministic transition systems have been considered before. In the work by \citet{LiuOza:16}, \LTLnn\ is used as the specification formalism, and the abstract transition systems are constructed from continuous-time dynamical systems subject to disturbances. \citet{NilOzaTopMur:12} considers abstractions of discrete-time systems subject to disturbances where $N$-step reachability is the basis for abstraction. However, for these approaches, controllers that exist in the concrete system might not have a corresponding controller in the abstract system.
Another abstraction approach is (approximate) \emph{alternating bisimulation}~\citep{AluHenKupVar:98, PolTab:09} which is an extension of (approximate) bisimulation  applicable for synthesis with \LTL\ specifications on non-deterministic systems.

This \whatsit\ addresses the problem of designing a controller for a given non-deterministic concrete transition system and an \LTLnn\ specification through an abstraction method that guarantees that a controller can be synthesised for the concrete system if and only if a controller can be synthesised for the abstract system.
By only considering \LTLnn\ specifications, it is possible to construct coarser abstractions compared to existing work.
The abstractions are based on the \emph{\rdsb} relation introduced in this \whatsit, which 
extends divergent stutter
bisimulation to non-deterministic systems, and relaxes alternating
bisimulation to \LTLnn\ specifications. As such, this work can be seen as a
combination of the works by \citet{MohMalWinLafOza:21} and
\citet{PolTab:09}. It is shown how an abstract transition system can be
constructed based on \rdsb, and how this abstraction is used to synthesise
and implement a controller for the concrete system.

In the following, \sect~\ref{sec:preliminaries} introduces notations of
transitions systems and linear temporal logic.
\sect~\ref{sec:abstraction} gives an overview of the proposed steps for
controller synthesis by abstraction, which are detailed in the
following \sects\
\ref{sec:robustStutterBisimulation}--\ref{sec:concreteController}.
Afterwards, \sect~\ref{sec:example} applies to idea to an example,
and \sect~\ref{sec:conclusions} adds concluding remarks.
Formal proofs of all results can be found in 
\iffull the appendix\else\citet{KroMalMohFab:22}\fi.

\section{Preliminaries}
\label{sec:preliminaries}

\subsection{Transition systems}
\label{sec:transsys}

Transition systems describe how systems or processes evolve from one state
to the next.
The states are labelled with atomic propositions, which can be used by
formulas to refer to specific states or groups of states.

\begin{definition}
  \label{def:ts}
  A \emph{transition system} is a tuple $G = \tsystem$ where
\iffull
  \begin{itemize}
  \item $S$ is a set of \emph{states};
  \item $\Sigma$ is a set of \emph{transition labels};
  \item $\delta \subseteq S \times \Sigma \times S$ is a
    \emph{transition relation};
  \item $S\init \subseteq S$ is a set of \emph{initial states};
  \item \AP\ is a set of \emph{atomic propositions};
  \item $L\colon S \to 2^\AP$ is a \emph{state labelling function}.
  \end{itemize}
\else
  $S$ is a set of \emph{states},
  $\Sigma$ is a set of \emph{transition labels},
  $\delta \subseteq S \times \Sigma \times S$ is a \emph{transition relation},
  $S\init \subseteq S$ is a set of \emph{initial states},
  \AP\ is a set of \emph{atomic propositions}, and
  $L\colon S \to 2^\AP$ is a \emph{state labelling function}.
\fi
\end{definition}

A system modelled by a difference equation
\begin{equation}
  x(t + \Delta t) = g(x(t), u(t), w(t))
\end{equation}
where $x \in \mathcal{X}$ is the state, $u \in \mathcal{U}$ is the control input, and  $w \in \mathcal{W}$ is the disturbance, can be represented as a transition system $G = \tsystem$ with state space $S = \mathcal{X}$ and transition labels $\Sigma = \mathcal{U}$ \citep{Tab:09}. If for $u \in \mathcal{U}$, there exists $w \in \mathcal{W}$ such that $s_1 = g(s_0, u, w)$, then there is transition $(s_0, u, s_1) \in \delta$.

Sequences of states of a transition system are used to represent traversals of the transition system. $S^*$~and $S^{\omega}$ denote the sets of all finite and infinite sequences of elements of~$S$, respectively, and $S^{\infty} = S^* \cup\mkern1.5mu S^{\omega}$. The set of non-empty finite sequences is denoted by $S^+ = S^* \setminus \{\varepsilon\}$, where $\varepsilon$ is the empty sequence. Two sequences $\rho \in S^*$ and $\pi \in S^{\infty}$ can be concatenated to form a new sequence $\rho \pi \in S^{\infty}$.
A finite sequence $\rho \in S^*$ is a \emph{prefix} of $\pi \in S^\infty$, written $\rho \sqsubseteq \pi$, if there exists a sequence $\pi' \in S^\infty$ such that $\rho \pi' = \pi$, and $\rho$ is a \emph{proper prefix} of~$\pi$, written $\rho \sqsubset \pi$, if $\rho \sqsubseteq \pi$ and $\rho \neq \pi$.
To be a traversal of a transition system, a sequence of states must be allowed by the transition relation:
\begin{definition}
  \label{def:Frags}
Let $G$ be a transition system.
A finite sequence of states $\rho = s_0 \cdots s_n \in S^*$ is a \emph{finite path fragment} of~$G$ if for all $0 \leq i < n$ there exists $\sigma_i \in \Sigma$ such that $(s_i, \sigma_i, s_{i+1}) \in \delta$. The set of all finite path fragments of~$G$ is denoted $\Frags^*(G)$.
An infinite sequence of states $\pi \in S^{\omega}$ is an \emph{infinite path fragment} of~$G$ if for all finite prefixes $\rho \sqsubset \pi$ it holds that $\rho \in \Frags^*(G)$. The set of all infinite path fragments of~$G$ is denoted $\Frags^\omega(G)$.
The set of all path fragments of~$G$ is $\Frags^\infty(G) = \Frags^*(G) \cup \Frags^\omega(G)$.
A \emph{path} of~$G$ is an infinite path fragment $\pi = s_0s_1\cdots \in \Frags^\omega(G)$ with $s_0 \in S\init$. The set of all paths of~$G$ is denoted $\Paths(G)$.
\end{definition}

That is, path fragments can start in any state of the transition system, whereas paths start in an initial state.

An important property of transition systems is \emph{deadlock freedom}, which ensures that progress is possible from every state.

\begin{definition}
\label{def:deadlockFree:G}
A transition system $G = \tsystem$ is \emph{\dlfree} if for each
state $s \in S$ there exists a label $\sigma \in \Sigma$ and a state $t \in
S$ such that $(s, \sigma, t) \in \delta$.
\end{definition}

\iffull
In a \dlfree system, every finite path fragment can be extended infinitely.
If a system is not \dlfree, it contains deadlock states from where no
further transitions are possible.
\fi

\subsection{Relations}

Given a set~$X$, a relation $\RR \subseteq X\times X$ is an \emph{equivalence relation} on~$X$ if it is reflexive, symmetric, and transitive. The \emph{equivalence class} of $x \in X$ is $[x]_\RR = \{\, x' \in X \mid (x, x') \in \RR \,\}$. The set of all equivalence classes modulo~\RR, the \emph{quotient space} of $X$ under \RR, is $X/\RR = \{\, [x]_\RR \mid x \in X \,\}$.
Such partitioning into equivalence classes is one way to abstract the state set of a transition system, where in the abstract system each equivalence class is one state.
A relation~$\RR_1$ is a \emph{refinement} of a relation~$\RR_2$ if $\RR_1 \subseteq \RR_2$. In this case, $\RR_1$ is said to be \emph{finer} than~$\RR_2$, and $\RR_2$ is \emph{coarser} than~$\RR_1$.

\begin{definition}\label{def:superblock}
  Let $\RR \subseteq X \times X$ be an equivalence relation. A set $T \subseteq X$ is a \emph{superblock} of~\RR, if for all $x_1 \in T$ and all $x_2 \in X$ such that $(x_1, x_2) \in \RR$, it holds that $x_2 \in T$. The set of all superblocks of~\RR\ is denoted $\SB(\RR)$.
\end{definition}

Superblocks are sets of elements that are closed under the equivalence
relation~\RR. Alternatively, they can be characterised as unions of zero
or more equivalence classes. Note that the empty set also is a superblock.

\begin{example}
\label{ex:equivalence}
Let $X = \{ 1, 2, 3, 4, 5\}$ and 
\begin{align*}
    \RR = \{ &(1, 1), (1, 2), (1, 3), (2, 1), (2, 2), (2, 3), (3, 1), \\
             &(3, 2), (3, 3), (4, 4), (4, 5), (5, 4), (5, 5) \}\ .
\end{align*}
\iffull
    $\RR$ is a relation on $X$ since $\RR \subseteq X \times X$, and it is an equivalence relation since it is reflexive, symmetric, and transitive.
    The equivalence classes are $[1]_\RR = [2]_\RR = [3]_\RR = \{ 1, 2, 3 \}$, and $[4]_\RR = [5]_\RR = \{ 4, 5 \}$. The
\else
    $\RR \subseteq X \times X$ is an equivalence relation, and the
\fi
set of equivalence classes modulo~\RR\ is $X/\RR = \{ [1]_\RR\bcom [2]_\RR\bcom [3]_\RR\bcom [4]_\RR\bcom [5]_\RR \} = \{ \{ 1, 2, 3 \}\bcom \{ 4, 5 \} \}$. The set of superblocks is $\SB(\RR) = \{ \emptyset\bcom \{ 1, 2, 3 \}\bcom \{ 4, 5 \}\bcom \{ 1, 2, 3, 4, 5 \} \}$.\qed
\end{example}

\subsection{Linear Temporal Logic}

A formula of \emph{Linear Temporal Logic without Next} (\LTLnn) is a logical formula consisting of atomic propositions from a set~$\AP$, the propositional logic operators, and the binary operator~$\until$.
Its syntax is defined by $\varphi = \top \mid p \mid \lnot \psi \mid \psi \land \theta \mid \psi \until \theta$, where $p \in \AP$, and $\psi$ and~$\theta$ are $\LTLnn$ formulas.

\begin{definition}\label{def:ltl}
Let $G = \tsystem$. Let $\psi$ and $\theta$ be $\LTLnn$ formulas, and let $p \in \AP$ be an atomic proposition. Whether a path fragment $\pi = s_0s_1\cdots \in \Frags^{\infty}(G)$ satisfies the $\LTLnn$ formula $\varphi$, written $\pi \vDash \varphi$, is defined inductively on the structure of~$\varphi$:
\begin{itemize}
    \item $\pi \vDash \top$ always holds;
    \item $\pi \vDash p$ iff $p \in L(s_0)$;
    \item $\pi \vDash \lnot \psi$ iff $\pi \models \psi$ does not hold;
    \item $\pi \vDash \psi \land \theta$ iff $\pi \vDash \psi$ and $\pi \vDash \theta$;
    \item $\pi \vDash \psi \until \theta$ iff there is $m \geq 0$ such that $s_ms_{m+1}\cdots \vDash \theta$ and for all $0 \leq i < m$ it holds that $s_is_{i+1}\cdots \vDash \psi$.
\end{itemize}
An \LTLnn\ formula~$\varphi$ holds at state~$s$, written $\langle G, s
\rangle \vDash \varphi$, if all infinite path fragments starting at~$s$
satisfy~$\varphi$, i.e, $\pi \models \varphi$ for all infinite path
fragments $\pi \in \Frags^\omega(G)$ with $s \sqsubset \pi$.
The transition system~$G$ \emph{satisfies} the
\LTLnn\ formula~$\varphi$, written $G \models \varphi$, if $\langle G,
s\init\rangle \models \varphi$ for all initial states $s\init \in S\init$.
\end{definition}

The other propositional logic operators ($\lor$, $\rightarrow$,
$\leftrightarrow$) can be defined based on $\top$, $\land$, and $\lnot$.
Some additional temporal operators can be defined based on the existing
ones by
$\finally \psi \equiv \top \until \psi$, 
$\globally \psi \equiv \lnot \finally \lnot \psi$, and 
$\psi \weakuntil \theta \equiv \globally \psi \lor (\psi \until \theta)$. The operators $\until$ and~$\weakuntil$, read as \emph{(strong) until} and \emph{weak until}, are of particular importance in this \whatsit. If $\psi \until \theta$ is satisfied on a path fragment, then $\theta$ eventually holds for some state in the path fragment, and $\psi$ holds in all states before that. $\psi \weakuntil \theta$ is similar to $\psi \until \theta$, but it is also satisfied on path fragments where $\psi$ holds in all states.

\LTLnn can be generalised to be defined over state sets $X\subseteq S$ instead of only atomic propositions. For a path fragment $\pi = s_0s_1\cdots \in \Frags^\infty(G)$, let $\pi \vDash X$ iff $s_0 \in X$. Such \LTLnn formulas will in the following be called \emph{generalised} $\LTLnn$ formulas.

\begin{definition}
Let $G=\tsystem$ be a transition system, and let $\RR \subseteq S \times S$ be an equivalence relation. A \emph{stutter step formula} for~$G$ is a generalised \LTLnn\ formula of the form $P \anyuntil T$, where $P,T \subseteq S$ with $P \cap T = \emptyset$ and $\anyuntil$ is either $\until$ or~$\weakuntil$.
If $P,T \in \SB(\RR)$, then the formula $P \anyuntil T$ is also called 
an \emph{\RR-superblock step formula} from~$P$.
\end{definition}

Stutter step formulas $P \anyuntil T$ describe the immediate future when a system is in some state set~$P$. The set~$T$ contains states that can be entered next from~$P$. A stutter step formula of the form $P \until T$ means that a path visiting~$P$ will eventually reach~$T$, whereas $P \weakuntil T$ means that it is possible to stay in~$P$ indefinitely or enter~$T$. Superblock step formulas require the source and target sets $P$ and~$T$ to be superblocks of an equivalence relation. They describe elementary steps of an abstract system whose states are equivalence classes.
The interest in stutter step formulas is because they make it possible to abstract away so-called \emph{stutter steps}, which are transitions within an equivalence class.
Whether a path fragment satisfies an \RR-superblock step formula is independent of the number of stutter steps. If \RR\ preserves state labels, then it can also be shown that the satisfaction of \LTLnn\ formulas is independent of such repeated steps. 

Given a transition system~$G=\tsystem$ and a sequence $\pi = s_0s_1\cdots \in \Frags^\infty$, the state labelling function~$L$ is extended to~$\pi$ by $L(\pi) = L(s_0)L(s_1)\cdots \in (2^{\AP})^\infty$. Moreover, the \emph{stutter free} sequence $\mathrm{sf}(\pi) \in S^\infty$ is obtained from~$\pi$ by removing all elements~$s_{i+1}$ such that $s_{i+1} = s_i$.

\begin{definition}\label{def:stutter:equivalent}
Let $G$ be a transition system. Two path fragments $\pi_1, \pi_2 \in \Frags^\infty(G)$ are \emph{stutter equivalent} if $\mathrm{sf}(L(\pi_1)) = \mathrm{sf}(L(\pi_2))$ and if~$\pi_1$ and~$\pi_2$ are either both finite or both infinite.
\end{definition}

\begin{theorem}[\citealt{BaiKat:08}]
	Let~$G$ be a transition system, $\varphi$ an \LTLnn\ formula, and let $\pi_1, \pi_2 \in \Frags^\omega(G)$ be two stutter equivalent path fragments. Then $\pi_1 \vDash \varphi$ iff $\pi_2 \vDash \varphi$.
\end{theorem}

\subsection{Controllers}

Transition systems can be controlled by restricting the subset of transitions that are allowed next. A controller does this by deciding, based on the history of visited states, the set of allowed transition labels.

\begin{definition}
Let $G = \tsystem$ be a transition
system. A \emph{general controller}, or simply \emph{controller}, for~$G$
is a function $C:S^+\to2^{\Sigma}$.
\end{definition}

When a controller is applied to a transition system~$G$, the resultant behaviour is described by another transition system, called the 
\emph{controlled} system. As the controller has access to all historic states of the original transition system~$G$, the states of the controlled system are sequences of states of~\G.

\begin{definition}
Let $G = \tsystem$ be a transition system, and let $C\colon S^+ \to 2^{\Sigma}$ be a controller for $G$. The \emph{controlled system} is $C/G = \langle S^+\bcom \Sigma\bcom \delta_C\bcom S\init\bcom \AP\bcom L_C\rangle$, where
\begin{equation*}
  \delta_C = \LongSet{5.5em}{$(s_0 \cdots s_n, \sigma, s_0\cdots s_n s_{n+1})
    \in S^+ \times \Sigma \times S^+$ $\mid$ $(s_n, \sigma, s_{n+1}) \in
    \delta$ and $\sigma \in C(s_0 \cdots s_n)$}.
\end{equation*}
and $L_C(s_0 \cdots s_n) = L(s_n)$.
\end{definition}

The initial states of the controlled system~$C/G$ are the initial states of the original system~$G$, interpreted as sequences of length one. After observing a state sequence $s_0\cdots s_n$,
the state of $G$ is~$s_n$ and the state of $C/G$ is $s_0\cdots s_n$.
The transitions from this state in~$C/G$ are those possible in~\G from~$s_n$ and allowed by~\C from $s_0\cdots s_n$.

\begin{definition}
  \label{def:deadlockFree:C}
  A controller~$C$ for a transition system~$G$ is \emph{\dlfree} if
  $C/G$ is \dlfree.
\end{definition}

\begin{definition}
\label{def:enforces}
Let $G$ be a transition system.
A controller~$C$ \emph{enforces} a (generalised)
\LTLnn-formula~$\psi$ from state~$s$ of~$G$ if $\langle C/G, s \rangle
\vDash \psi$, and $C$ enforces~$\psi$ on~$G$ if $C/G\vDash\psi$.
\end{definition}

In other words, an \LTLnn\ formula is enforced from a state if every controlled path fragment starting from that state satisfies the formula. The state in question may or may not be an initial state, as controllers may also be defined for path fragments that do not start with an initial state. An \LTLnn\ formula is enforced on a transition system if it is enforced from every initial state.

The path fragments of a controlled system~$C/G$ are given by $\Frags^\infty(C/G)$, which by definition are sequences of states of~$C/G$ and thus sequences of sequences of states of~$G$. The following definition of \emph{permitted} path fragments, $\Frags^\infty(C,G)$, is used to replace such by simple sequences of states, projecting a path fragment $(s_0) \penalty750(s_0 s_1) \cdots \penalty750(s_0 s_1 \cdots s_n) \cdots$ of the controlled system~$C/G$ to the path fragment $s_0 s_1 \cdots s_n \cdots$ of~$G$.

\begin{definition}
	\label{def:permitted}
	Let $G$ be a transition system, and
	let $C\colon S^+ \to 2^\Sigma$ be a controller for~$G$.
	A finite path fragment $\rho = s_0 s_1 s_2 \cdots s_n \in \Frags^*(G)$
	is \emph{permitted} by~$C$ if for all $0 \leq i < n$ there exists $\sigma_i
	\in C(s_0\cdots s_i)$ such that $(s_i,\sigma_i,s_{i+1}) \in \delta$.
	The set of all finite path fragments in~$G$ permitted by~$C$ is denoted by
	$\Frags^*(C,G)$.
	The definition is extended to permitted infinite path fragments, $\Frags^\omega(C,G)$, and
	permitted path fragments, $\Frags^\infty(C,G)$. The set of permitted paths, $\Paths(C, G)$, is defined likewise.
\end{definition}

\subsection{Positional Controllers}

\begin{definition}
Let $G = \tsystem$ be a transition system. A \emph{\memoryless\ controller}
for~$G$ is a function $\overlinit{C}\colon S \to 2^{\Sigma}$.
\end{definition}

\begin{definition}
Let $G = \tsystem$ be a transition system,
and let $\overlinit{C}\colon S \to 2^{\Sigma}$ be a \memoryless\ controller for~$G$.
The \emph{controlled system} of $G$ under the control of~$\overlinit{C}$ is $\overlinit{C}/G = C/G$ where $C\colon S^+ \to 2^\Sigma$ with $C(\rho u) = \overlinit{C}(u)$ for all $\rho \in S^*$ and $u \in S$.
\end{definition}

As a \memoryless\ controller only makes decisions based on the current state, its controlled system can be alternatively defined using the original transition system state space.

Later, it will be useful to refer to the sets of states of a transition system from which there exists some (\memoryless) controller that can enforce a given generalised $\LTLnn$ formula.

\begin{definition}
  \label{def:EC}
  Let $G$ be a transition system, and let $\psi$ be a generalised \LTLnn-formula for~$G$. The set of states where $\psi$ can be
  enforced by control is:
  \begin{align*}
    \kern-1em
    \EC[G]\psi &= \LongSet{7.5em}{$s \in S$ $\mid$ there is some
      \dlfree controller~$C$ such that $\langle C/G, s \rangle \models \psi$}; \\
    \kern-1em
    \ECS[G]\psi &= \LongSet{7.5em}{$s \in S$ $\mid$ there is some \dlfree \memoryless\
      controller~$\smash{\overlinit{C}}$ such that $\langle \overlinit{C}/G, s \rangle \models \psi$}.
  \end{align*}
\end{definition}

\begin{figure}
  \centering
  \begin{tikzpicture}
    \node [state] (s0) {0};
    \node [state] (s1) [left=of s0] {$1:a$};
    \node [state] (s2) [right=of s0] {$2:b$};
    \path [->] (s0) edge [bend right] node [above] {$\alpha$} (s1);
    \path [->] (s1) edge [bend right] node [below] {$\alpha$} (s0);
    \path [->] (s0) edge [bend left] node [above] {$\beta$} (s2);
    \path [->] (s2) edge [bend left] node [below] {$\beta$} (s0);
  \end{tikzpicture}
  \caption{Example showing the difference between \memoryless\ and general
    controllers. From state~0, the \LTLnn\ specification $\finally a \land
    \finally b$ can only be enforced by a
general controller.}
  \label{fig:ECvsECS}
\end{figure}
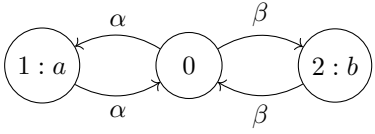

\begin{example}
\label{ex:ECvsECS}
Consider the transition system $G = \tsystem$ in \fig~\ref{fig:ECvsECS}.
\iffull
Here
\begin{itemize}
    \item $S = \{0,1,2\}$, 
    \item $\Sigma = \{\alpha,\beta\}$, 
    \item $\delta = \{(0,\alpha,1)\bcom (1,\alpha,0)\bcom (0,\beta,2)\bcom (2,\beta,0)\}$,
    \item $S\init = \{0\}$, 
    \item $\AP = \{a,b\}$, 
    \item $L(0) = \emptyset$, $L(1) = \{a\}$, and $L(2) = \{b\}$.
\end{itemize}
\else
It has three states, $S = \{0,1,2\}$, labelled by $L(1) = \{a\}$ and
$L(2) = \{b\}$.
\fi

The \LTLnn\ specification $\varphi_1 \equiv \finally a$ can
be enforced by a \memoryless\ controller $\overlinit{C}_1$ with $\overlinit{C}_1(0) = \overlinit{C}_1(1) = \{\alpha\}$ and $\overlinit{C}_1(2) = \{\beta\}$.
Thus $\overlinit{C}_1/G \vDash \varphi_1$.
It is also clear that $\ECS[G]{\varphi_1} = \EC[G]{\varphi_1} = S$.

However, the specification $\varphi_2 \equiv \finally a \land \finally b$ cannot be enforced from the state~0 by any \memoryless\ controller.
This is because a \memoryless\ controller~$\overlinit{C}$ must choose the same control action each time state~0 is entered.
So, either $\overlinit{C}(0) = \{\alpha\}$ or $\overlinit{C}(0) = \{\beta\}$ or $\overlinit{C}(0) = \{\alpha,\beta\}$.
In the first case, $\overlinit{C}$ permits an infinite cycle between states 0 and~1, failing $\finally b$. In the second case, $\overlinit{C}$ permits an infinite cycle between states 0 and~2, failing $\finally a$.
In the third case, $\overlinit{C}$ permits both these cycles and thus fails both $\finally a$ and~$\finally b$ from state~0.
Thus $0 \notin \ECS[G]{\varphi_2}$.

From states 1 and~2, on the other hand, there exist \dlfree \memoryless\ controllers that enforce~$\varphi_2$, e.g.,
$\overlinit{C}_2$ with $\overlinit{C}_2(0) = \overlinit{C}_2(1) = \{\beta\}$ and $\overlinit{C}_2(1) = \{\alpha\}$ enforces $\varphi_2$ from state~1. Hence, $1 \in \ECS[G]{\varphi_2}$. Likewise, $2 \in \ECS[G]{\varphi_2}$ and therefore $\ECS[G]{\varphi_2} = \{1,2\}$. 

As a general 
controller $C : S^+ \to 2^{\Sigma}$ can use the history of visited states to determine its control actions, it can enforce more \LTLnn specifications than a \memoryless\ controller. \C defined as
\begin{align*}
    C(0) &= \{ \beta \} \\
    C(s_0 \cdots s_{n-1} s_n) &= \left\{
    \begin{array}{ll}
        \{ \alpha \} & \textrm{if } s_{n-1} = 2 \\
        \{ \beta \} & \textrm{if } s_{n-1} = 1 \\
    \end{array}
    \right.
\end{align*}
enforces $\varphi_2$ from all states, so $\EC[G]{\varphi_2} = S$.\qed
\end{example}

It is clear that $\ECS[G]\psi \subseteq \EC[G]\psi$ for every
generalised \LTLnn\ formula~$\psi$ (a \memoryless\ controller is a special case of a general controller), but the opposite inclusion does not hold in general as shown in \Examp~\ref{ex:ECvsECS}.
However, \appref{as shown in \app~\ref{app:fixpoints}, }for stutter step formulas, i.e., generalised \LTLnn formulas of the form $P \until T$ and $P \weakuntil T$, the reverse inclusion does indeed hold. 

\begin{proposition}
  \label{prop:EC:any}
  Let $G = \tsystem$ be a \dlfree transition system, and let $\psi$
  be a stutter step formula for~$G$. Then $\EC[G]\psi = \ECS[G]\psi$, and
  there exists a \memoryless\ controller $\overlinit{C}\colon S \to 2^\Sigma$ such that,
  for all states $s \in \ECS[G]\psi$ it holds that $\langle \overlinit{C}/G,s\rangle \models \psi$.
\end{proposition}

\appref{The proof for this result is by combination of \props\ \ref{prop:EC:W} and~\ref{prop:EC:U} of \app~\ref{app:fixpoints}, which also explains how to compute $\ECS[G]{\psi}$ for stutter step formulas.}

\section{Controller Synthesis by Abstraction}
\label{sec:abstraction}

The problem considered in this \whatsit\ concerns synthesis of a controller enforcing
an \LTLnn\ formula~$\varphi$ for a concrete transition system~$G$ that models a discrete-time dynamical system, as illustrated by the dashed arrow in \fig~\ref{fig:abstraction:principle}. 
Such discrete-time systems are often derived from difference equations with infinite state, control, and disturbance spaces, so that the transition system's state space becomes infinite.

\begin{namedproblem}{Synthesis}\label{problem:synthesis}
  Given a transition system~$G$ and an \LTLnn\ formula~$\varphi$, find a
  \dlfree controller~$C$ that enforces $\varphi$ on~$G$.
\end{namedproblem}

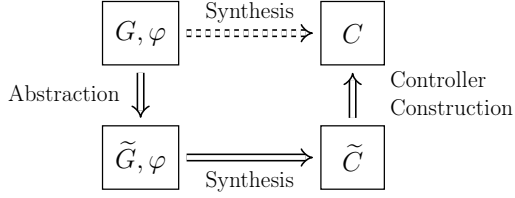
\begin{figure}
\centering
\scalebox{0.6}{	
\tikzset{
  block/.style = {draw, thick, rectangle, minimum height = 4em, minimum width = 4em},
}
\begin{tikzpicture}[every node/.style={outer sep=5pt, inner sep=8pt}]
  \node [block] (G) {\Large$G, \varphi$};
  \node [block, right=28mm of G] (CG) {\Large$C$};
  \node [block, below=of G] (GR) {\Large$\Tilde{G}, \varphi$};
  \node [block] (CR) at (CG |- GR) {\Large$\Tilde{C}$};
		
  \draw [-{Implies}, double distance=3pt, line width=1pt, dashed] (G) -- node[above=-1.5ex] {\large Synthesis} (CG);
  \draw [-{Implies}, double distance=3pt, line width=1pt] (G) --
  node[left] {\large Abstraction\strut} (GR);
  \draw [-{Implies}, double distance=3pt, line width=1pt] (GR) -- node[below=-1ex] {\large Synthesis} (CR);
  \draw [-{Implies}, double distance=3pt, line width=1pt] (CR) --
  node[right] {\large \renewcommand{\arraystretch}{1}
    \begin{tabular}{l} Controller \\ Construction \end{tabular}} (CG);
\end{tikzpicture}
} 
\caption{Illustration of controller synthesis by abstraction. The dashed arrow represents synthesis directly on the concrete system, the solid arrows represent synthesis by abstraction.}
\label{fig:abstraction:principle}
\end{figure}

To apply methods for finite transition systems to infinite state spaces, this \whatsit\ uses an approach based on abstraction
shown by the solid arrows in \fig~\ref{fig:abstraction:principle}.
First, the concrete transition system~$G$ is transformed into a finite-state
\emph{abstraction}~$\Tilde{G}$. Second, an \emph{abstract controller}~$\Tilde{C}$, which 
enforces~$\varphi$ on the abstraction~$\Tilde{G}$, is synthesised.
Third, a \emph{concrete controller}~$C$
is constructed based on the abstract controller~$\Tilde{C}$, which then enforces $\varphi$ on the concrete system~$G$.
As the synthesis of controllers for \LTLnn\ specifications is addressed
elsewhere \citep{Ram:89,KloBel:08}, this \whatsit\ is only concerned with the
first and third steps, i.e., construction of the abstraction~$\Tilde{G}$ and the
concrete controller~$C$.

\begin{namedproblem}{Abstraction}\label{problem:abstraction}
	Given a concrete transition system~$G$, and an \LTLnn\ formula~$\varphi$, construct an abstract transition system~$\Tilde{G}$ such that, there exists a \dlfree controller~$\Tilde{C}$ enforcing~$\varphi$ on $\Tilde{G}$ iff there exists a \dlfree controller~$C$ enforcing~$\varphi$ on~$G$.
\end{namedproblem}

Solving the \ref{problem:abstraction} means that the abstract transition system~$\Tilde{G}$ must be constructed in such a way that it is equivalent to the concrete system~$G$ with respect to synthesis. This ensures the \emph{soundness} and \emph{completeness} of the approach, i.e., every controller obtained from the abstract system is related to a controller for the concrete system, and if there exists a controller for the concrete system, then it can be found by synthesis based on the abstract system.

Even though the existence of a controller~$\Tilde{C}$ for the abstract system carries over to the existence of a controller~$C$ for the concrete system, it is not immediately clear how to construct such a concrete controller. The last problem is to find an effective way to construct the concrete controller~$C$ from the abstract controller~$\Tilde{C}$. 

\begin{namedproblem}{Controller Construction}\label{problem:construction}
	Given a concrete transition system~$G$, an abstract transition system~$\Tilde{G}$, an \LTLnn\ formula~$\varphi$, and a \dlfree controller~$\Tilde{C}$ that enforces~$\varphi$ on~$\Tilde{G}$, construct a \dlfree controller~$C$ that enforces~$\varphi$ on~$G$.
\end{namedproblem}

In the following, \Sect~\ref{sec:robustStutterBisimulation} introduces the robust stutter bisimulation relation, which is used in \Sect~\ref{sec:quotient} to define a quotient transition system~$\Tilde{G}$ that solves the \ref{problem:abstraction}. \Sect~\ref{sec:algorithm} describes an algorithm to compute the needed robust stutter bisimulation relation. Finally, assuming a controller~$\Tilde{C}$ has been synthesised for the abstraction, \Sect~\ref{sec:concreteController} describes the construction of a concrete controller~$C$ and solves the \ref{problem:construction}.

\section{Robust stutter bisimulation}
\label{sec:robustStutterBisimulation}

This section defines the \emph{robust stutter bisimulation} relation,
which identifies two transition systems as equivalent if the same
\LTLnn\ formulas can be enforced on both systems. This is the crucial
criterion to identify solutions to the \ref{problem:abstraction},
ensuring that any transition system~$\Tilde{G}$ that is robust stutter
bisimilar to the concrete system~$G$ can be used as its abstraction, and
thus a solution to the \ref{problem:abstraction}. However, the relation will first be defined over states of one transition system.

The goal is to classify two states of a transition system as robust stutter bisimilar when the same \LTLnn\ formulas can be enforced from both.
Due to the absence of the ``next'' operator, this condition can be
simplified by considering only stutter step formulas, or more precisely
superblock step formulas, as will be clear later. Given an equivalence class~$P$, a superblock step
formula $P \until T$ or $P \weakuntil T$ expresses that the system can
transition to other equivalence classes with or without the guarantee that this transition occurs eventually. The condition can be further simplified to superblock step formulas that can be enforced by \memoryless\ controllers because of \Propn~\ref{prop:EC:any}.
Controllers for arbitrary \LTLnn\ formulas can be constructed by combining \memoryless\ controllers enforcing stutter step formulas from different states.

\begin{definition}
\label{def:robust:stutter:bisimulation}
Let $G = \tsystem$ be a transition system. An equivalence relation $\RR \subseteq S \times S$ is a \emph{robust stutter bisimulation} on~$G$ if for all $(s_1, s_2) \in \mathcal{R}$ the following conditions hold:
\begin{enumerate}
    \item $L(s_1) = L(s_2)$ \label{it:rsb:label}
    \item \label{it:rsb:equivalent}
    for every \RR-superblock step formula $\psi$ from~$[s_1]_\RR$ such that $s_1 \in \ECS[G]{\psi}$ it also holds that $s_2 \in \ECS[G]{\psi}$.
\end{enumerate}
\end{definition}

In other words, a given relation~\RR\ is a robust stutter bisimulation if \ref{it:rsb:label} it preserves the labels of states, and \ref{it:rsb:equivalent} precisely the same superblock step formulas can be enforced from equivalent states.
Although condition~\ref{it:rsb:equivalent} is not inherently symmetric, symmetry is ensured as \RR\ is required to be an equivalence relation. By \propn~\ref{prop:EC:any} it is enough to define enforceability of a formula~$\psi$ in a state~$s$ through $\ECS[G]\psi$, i.e., using \memoryless\ controllers.
\Defn~\ref{def:robust:stutter:bisimulation} can be considered as a generalisation of \emph{divergent stutter bisimulation} \citep{BaiKat:08} to transition systems with disturbances.

\begin{figure}
	\centering
	\begin{tikzpicture}[>=latex, initial text={}, initial where=above, node distance=8mm and 7.5mm]
		\node [state] (a1) {$a_1:a$};
		\node [state] (a2) [right=of a1] {$a_2:a$};
		\node [draw=black, dashed, fit=(a1) (a2), label=left:$\Aq_1$] {};
		\node [state] (a3) [below=of a2] {$a_3:a$};
		\node [state] (a4) at (a1 |- a3) {$a_4:a$};
		\node [draw=black, dashed, fit=(a3) (a4), label=left:$\Aq_2$] {};
		
		\node [state] (b4) [right=8mm of a3] {$b_4:b$};
		\node [state] (b3) [right=of b4] {$b_3:b$};
		\node [state] (b1) at (a2 -| b4) {$b_1:b$};
		\node [state] (b2) at (b1 -| b3) {$b_2:b$};
		\node [draw=black, dashed, fit=(b3) (b4), label=right:$\Bq_{2}$] {};
		\node [draw=black, dashed, fit=(b1) (b2), label=right:$\Bq_1$] {};
		
		\node [state] (c1) [below=of b4] {$c_1:c$};
		\node [state] (c2) at (b3 |- c1) {$c_2:c$};
		\node [draw=black, dashed, fit=(c1) (c2), label=right:$\Cq$] {};
		
		\node [state] (d2) at (a3 |- c1) {$d_2:d$};
		\node [state, initial] (d1) at (a4 |- d2) {$d_1:d$};
		\node [draw=black, dashed, fit=(d1) (d2), label=left:$\Dq$] {};
		
		\path [->] (a1) edge [bend right] node [above] {$\sigma_1$} (a2);
		\path [->] (a2) edge [bend right] node [below] {$\sigma_1$} (a1);
		
		\path [->] (a4) edge [bend right] node [above] {$\sigma_1$} (a3);
		\path [->] (a3) edge [bend right] node [below] {$\sigma_1$} (a4);
		
		\path [->] (a3) edge [bend right] node [above] {$\sigma_2$} (b4);
		
		\path [->] (b2) edge node [below] {$\sigma_1$} (b1);
		
		\path [->] (b1) edge node [above] {$\sigma_1$} (a2);
		\path [->] (b1) edge node [above=1mm] {$\sigma_1$} (a3);
		
		\path [->] (b3) edge node [above] {$\sigma_1$} (b4);
		
		\path [->] (b3) edge node [left] {$\sigma_2$} (b2);
		\path [->] (b4) edge node [right] {$\sigma_2$} (b1);
		
		\path [->] (b4) edge [bend right] node [below] {$\sigma_1$} (a3);
		\path [->] (b4) edge [bend right] node [left] {$\sigma_1$} (c1);
		
		\path [->] (c1) edge [bend right] node [right] {$\sigma_1$} (b4);
		\path [->] (c2) edge node [left] {$\sigma_1$} (b3);
		
		\path [->] (d1) edge [bend right] node [above] {$\sigma_1$} (d2);
		\path [->] (d2) edge [bend right] node [below] {$\sigma_1$} (d1);
		
		\path [->] (d2) edge node [above] {$\sigma_1$} (c1);
		\path [->] (d2) edge node [left] {$\sigma_2$} (a3);
	\end{tikzpicture}
	\caption{A transition system~$G$ with a robust stutter bisimulation~\RR.}
	\label{fig:quotient:system:concrete}
\end{figure}
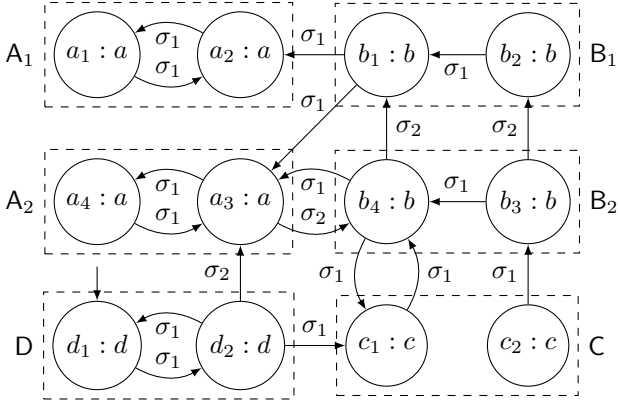

\begin{example}\label{ex:robust:stutter:bisimulation}
  Consider the transition system $G = \tsystem$ in \fig~\ref{fig:quotient:system:concrete} and the equivalence relation $\RR$ that divides the state space into six equivalence classes $S/\RR = \{ \Aq_1, \Aq_2, \Bq_1, \Bq_2, \Cq, \Dq \}$ where $\Aq_1 = \{ a_1\bcom a_2 \}$, $\Aq_2 = \{ a_3\bcom a_4 \}$, $\Bq_1 = \{ b_1, b_2 \}$, $\Bq_2 = \{ b_3, b_4 \}$, $\Cq = \{c_1, c_2\}$, and $\Dq = \{d_1, d_2\}$. Clearly, \RR\ fulfils condition~\ref{it:rsb:label} in \defn~\ref{def:robust:stutter:bisimulation} since the states in each class are all labelled by the same label from $\AP = \{ a, b, c, d \}$. By enumerating all \dlfree \memoryless\ controllers and all \RR-superblock step formulas it can also be verified that condition~\ref{it:rsb:equivalent} is fulfilled by \RR. For instance, the \dlfree \memoryless\ controller $\overlinit{C}(s) = \{ \sigma_1 \}$ enforces the \RR-superblock step formula $\Dq \weakuntil (\Cq \cup \Aq_2)$ from all states $d_i \in \Dq$, it enforces $\Aq_1 \weakuntil \emptyset$ from all $a_i \in \Aq_1$. The formula $\Aq_1 \weakuntil \emptyset$ is equivalent to $\globally\, \Aq_1$ and expresses that the system remains in the equivalence class~$\Aq_1$ indefinitely.
  
  Consider instead the equivalence relation~$\RR'$ where $S/\RR' = \{ \Aq_1, \Aq_2, \Bq_1 \cup \Bq_2, \Cq, \Dq \}$.
  This relation does not satisfy condition~\ref{it:rsb:equivalent}, which can be seen as follows. The above controller $\overlinit{C}$ enforces the $\RR'$-superblock step formula $\psi \equiv (\Bq_1 \cup \Bq_2) \until (\Aq_2 \cup \Cq)$ from the states $b_3$ and~$b_4$, so it follows that $b_3, b_4 \in \ECS{\psi}$. On the other hand, to be \dlfree, any \memoryless\ controller for~$G$ must enable $\sigma_1$ from~$b_1$, enabling the transition from $b_1$ to $a_2 \notin \Aq_2 \cup \Cq$, which means that $\psi$ cannot be enforced from~$b_1$, i.e., $b_1 \notin \ECS{\psi}$. Thus, condition \ref{it:rsb:equivalent} of \defn~\ref{def:robust:stutter:bisimulation} is violated for $(b_1,b_4) \in \RR'$ and $\RR'$ is not a robust stutter bisimulation.\qed
\end{example}

While \defn~\ref{def:robust:stutter:bisimulation} identifies whether a given relation is a robust stutter bisimulation, it is also of interest to define two states as robust stutter bisimilar. This is the case if there exists a robust stutter bisimulation where they are equivalent.

\begin{definition}
\label{def:inbisim}
Let $G$ be a transition system. States $s_1, s_2 \in S$ are \emph{robust 
stutter bisimilar}, denoted $s_1 \bisim
s_2$, if there exists a robust stutter bisimulation relation $\RR$ on~$G$
with $(s_1, s_2) \in \RR$.
\end{definition}

\appref{This defines a relation $\inbisim \subseteq S \times S$ that includes every robust stutter bisimulation. As shown in \app~\ref{app:coarsest}, this relation is again a robust stutter bisimulation. Therefore \inbisim\ is also called the \emph{coarsest robust stutter bisimulation} on~$G$.}

To identify two transition systems as equivalent, the definition of a robust stutter bisimulation on a single transition system is lifted to a relation between transition systems. A way to achieve this is by considering the \emph{union} transition system.

\begin{definition}
Let $G_1 = \tsystem[1]$ and $G_2 = \tsystem[2]$ be two transition systems with $S_1 \cap S_2 = \emptyset$. The \emph{union} of $G_1$ and~$G_2$ is
\iffull
\begin{equation*}
  G_1 \cup G_2 = \tsystem \ ,
\end{equation*}
\else
$G_1 \cup G_2 = \tsystem$,
\fi
where $S = S_1 \cup S_2$, $\Sigma = \Sigma_1 \cup \Sigma_2$, $\delta = \delta_1 \cup \delta_2$, $S\init = S\init_1 \cup S\init_2$, $\AP = \AP_1 \cup \AP_2$, $L(s) = L_1(s)$ for $s \in S_1$, and $L(s) = L_2(s)$ for $s \in S_2$.
\end{definition}

\begin{definition}\label{def:robust:stutter:between}
Let $G_1 = \tsystem[1]$ and $G_2 = \tsystem[2]$ be two transition systems with $S_1 \cap S_2 = \emptyset$. A relation $\RR \subseteq S_1 \times S_2$ is a robust stutter bisimulation between $G_1$ and~$G_2$, if \RR\ is a robust stutter bisimulation on $G_1 \cup G_2$ and the following conditions hold:
\begin{align*}
  \forall s_1 \in S_1\init \ldotp \exists s_2 \in S_2\init \ldotp (s_1, s_2) \in \mathcal{R} \ , & \quad\text{and} \\
  \forall s_2 \in S_2\init \ldotp \exists s_1 \in S_1\init \ldotp (s_1, s_2) \in \mathcal{R} \ .
\end{align*}
$G_1$ and~$G_2$ are robust stutter bisimilar, $G_1 \approx G_2$, if there exists a
robust stutter bisimulation between $G_1$ and~$G_2$.
\end{definition}

\section{Quotient transition system}
\label{sec:quotient}

Given a robust stutter bisimulation~\RR\ on the states of a transition system, a solution to
the \ref{problem:abstraction} still requires to construct an equivalent and hopefully smaller transition system. This is typically done by constructing a \emph{quotient} transition system whose states are the equivalence classes induced by the equivalence relation~\RR. Standard quotient constructions \citep{BaiKat:08} fail to ensure robust stutter bisimilarity between the original and quotient systems because they do not distinguish the different stutter step formulas. This is solved by the following alternative definition.

\begin{definition}\label{def:quotient:system}
Let $G = \tsystem$ be a transition system, and let \RR\ be a robust stutter bisimulation on~$G$. The \emph{robust stutter bisimulation quotient} is
\begin{equation}
  G/\RR = \langle S/\RR, \Sigma_{\RR}, \delta_{\RR}, S\init_{\RR}, \AP, L_{\RR} \rangle 
\end{equation}
where
\begin{align}
\Sigma_{\RR} &= \LongSet{6em}{$\psi$ $\mid$ $\psi$ is
  an \RR-superblock step formula from $P \in S/\RR$ such that $P \subseteq
  \ECS[G]\psi$}; \label{eq:quotient:Sigma} \\
\delta_\RR &= \LongSet{6em}{$(P, \psi, T') \in S/\RR \times \Sigma_\RR
  \times S/\RR$ $\mid$ $\psi \equiv P \anyuntil T$ such that $T' \in S/\RR$
  and $T' \subseteq T$}{$\cup$} \label{eq:quotient:delta} \\
&\hphantom{{}={}}  \LongSet{6em}{$(P, \psi, P) \in S/\RR \times \Sigma_\RR \times S/\RR$ $\mid$ $\psi \equiv P \weakuntil T$ for some $T \in \SB(\RR)$}; \nonumber \\
S\init_{\RR} &= \{\, [s\init]_{\RR} \mid s\init \in S\init \,\} \ ; \label{eq:quotient:Sinit}
\end{align}
and $L_{\RR}([s]_{\RR}) = L(s)$ for all $s \in S$.
\end{definition}

The transition labels~$\Sigma_\RR$ are \RR-superblock step formulas of the form $P \anyuntil T$ that can be enforced from the source set~$P$ in the concrete system~$G$. The quotient transitions~$\delta_\RR$ are then defined based on the formulas these labels represent. A label $P \anyuntil T$, where $P$ is an equivalence class and $T$ is a superblock, is attached to transitions from $P$ to each equivalence class (each abstract state) that constitutes~$T$, representing the controller's ability to force the system to~$T$ without being able to determine the precise concrete state entered. A label $P \weakuntil T$ is additionally attached to a \selfloop\ transition on~$P$, representing the fact that the system may also stay in~$P$.

\begin{figure}
\centering
\begin{tikzpicture}[>=latex, initial text={}, initial where=left]
	\node [state] (A1) {$\Aq_1\!:a$};
	\node [state] (B1) [right=25mm of A1] {$\Bq_1\!:b$};
	\node [state] (A2) [below=of A1] {$\Aq_2\!:a$};
	\node [state] (B2) at (A2 -| B1) {$\Bq_2\!:b$};
	\node [state] (C) [below=of B2] {$\Cq:c$};
	\node [state, initial] (D) at (A2 |- C) {$\Dq:d$};
	
	\path [->] (A1) edge [loop left] node [left] {\small $\Aq_1 \weakuntil \emptyset$} (A1);
	\path [->] (B1) edge node [below] {{\small $\Bq_1 \until (\Aq_1 \cup \Aq_2)$}} (A1);
	\path [->] (B1) edge node [left=2mm] {\small $\Bq_1 \until (\Aq_1 \cup \Aq_2)$} (A2);
	\path [->] (A2) edge [loop left] node [left] {\small $\Aq_2 \weakuntil \emptyset$} (A2);
	\path [->] (A2) edge [bend right] node [above] {\small $\Aq_2 \until \Bq_2$} (B2);
	\path [->] (B2) edge node [right] {\small $\Bq_2 \until \Bq_1$} (B1);
	\path [->] (B2) edge node [above] {\small\quad $\Bq_2 \until (\Aq_2 \cup \Cq)$} (A2);
	\path [->] (B2) edge [bend right] node [below left] {\small $\Bq_2 \until (\Aq_2 \cup \Cq)$} (C);
	\path [->] (C) edge [bend right] node [right] {\small $\Cq \until \Bq_2$} (B2);
	\path [->] (D) edge [in=110, out=140, looseness=8] node [left=1mm] {\small $\Dq \weakuntil \Cq$} (D);
	\path [->] (D) edge node [below] {\small $\Dq \weakuntil \Cq$} (C);
	\path [->] (D) edge node [pos=0.6, right] {\small $\Dq \until \Aq_2$} (A2);
	
	\path [->] (A1) edge [in=120, out=150, looseness=8, dashed] node [pos=0.4, left] {\small $\Aq_1 \weakuntil \Bq_1$} (A1);
	\path [->] (A1) edge [bend left, looseness=0.7, dashed] node [above] {\small $\Aq_1 \weakuntil \Bq_1$} (B1);
\end{tikzpicture}
\caption{A part of the quotient system $G/\RR$.}
\label{fig:quotient:system:abstract}
\end{figure}
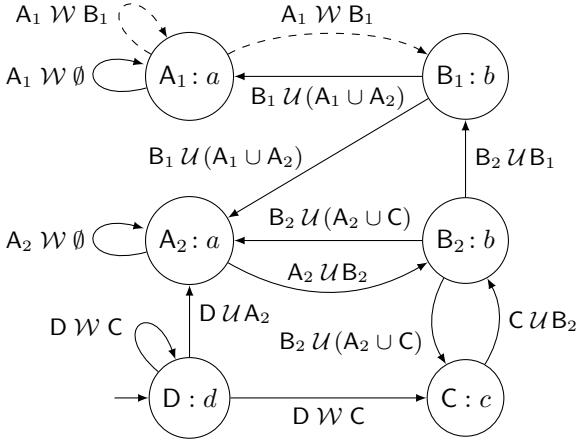

\begin{example}\label{ex:quotient:system}
	Consider again the transition system~$G$ and the robust stutter bisimulation~\RR\ in \fig~\ref{fig:quotient:system:concrete}. Its quotient system~$G/\RR$ is given in \fig~\ref{fig:quotient:system:abstract}. The states are the six equivalence classes $\Aq_1$, $\Aq_2$, $\Bq_1$, $\Bq_2$, $\Cq$, and~$\Dq$, and their labels $a$, $b$, $c$, and~$d$ match the corresponding concrete states.
	The transition labels in $\Sigma_\RR$ are constructed based on the existence of \memoryless\ controllers. As seen in \Examp~\ref{ex:robust:stutter:bisimulation}, the \memoryless\ controller $\overlinit{C}(s) = \{\sigma_1\}$ enforces different \RR-superblock step formulas from different equivalence classes. For instance, $\langle \overlinit{C}/G, a_i\rangle \vDash \Aq_1 \weakuntil \emptyset$ for $a_i \in \Aq_1$, so $\Aq_1 \subseteq \ECS[G]{\Aq_1 \weakuntil \emptyset}$. Clearly, $\Aq_1 \in S/\RR$, so $\psi_1 \equiv \Aq_1 \weakuntil \emptyset \in \Sigma_\RR$ by~\eqref{eq:quotient:Sigma}. Other formulas such as $\Aq_1 \weakuntil \Bq_1$ are implied by~$\psi_1$ and can also be enforced, so a total of 31 formulas $\Aq_1 \weakuntil T$ with $T \subseteq \{\Aq_2,\Bq_1,\Bq_2,\Cq,\Dq\}$ are included in $\Sigma_\RR$ (not all shown in the figure). According to~\eqref{eq:quotient:delta}, the label $\Aq_1 \weakuntil \emptyset$ is attached to the \selfloop\ transition $\Aq_1 \to \Aq_1$, and $\Aq_1 \weakuntil \Bq_1$ is attached to $\Aq_1 \to \Aq_1$ and $\Aq_1 \to \Bq_1$, etc. On the other hand, no stutter step formula $\Aq_1 \until T$ is enforceable from~$\Aq_1$, so those formulas do not appear in~$\Sigma_\RR$. 
	From state~$d_2$, it is not possible to force a transition to only one of $d_1$ or~$c_1$, but it holds that $\langle \overlinit{C}/G, d_i \rangle \vDash \Dq \weakuntil \Cq$. Hence, $\Dq \weakuntil \Cq  \in \Sigma_\RR$ along with implied formulas.
	\Fig~\ref{fig:quotient:system:concrete} shows solid arrows for all transitions corresponding to the labels formed by the strongest \RR-superblock step formulas (in terms of implication) that can be enforced.
	\qed
\end{example}

The following theorem confirms that the quotient $G/\RR$ of a concrete transition system~$G$ with respect to a robust stutter bisimulation~\RR\ is robust stutter bisimilar to the concrete system~$G$. \appref{The proof is given in \app~\ref{app:quotient}.}

\begin{theorem}\label{thm:quotient:robust:stutter:bisimilar}
Let $G$ be a transition system, and let \RR\ be a robust stutter bisimulation on~$G$. Then $G \bisim G/\RR$.
\end{theorem}

\section{Computing the quotient partition}
\label{sec:algorithm}

As shown in \sects~\ref{sec:robustStutterBisimulation}
and~\ref{sec:quotient}, the \ref{problem:abstraction} can be solved  by constructing a quotient with respect to a robust
stutter bisimulation~\RR. While such a relation~\RR\ always exists (e.g.,
the identity relation), it is not immediately clear how to obtain a robust
stutter bisimulation that results in a small quotient. This section
presents an algorithm to compute the coarsest robust stutter bisimulation
on a given transition system, for which the quotient has the smallest
number of states possible.
The idea is to compute the desired relation by \emph{partition refinement}
\citep{PaiTar:87}, where an initial partition is repeatedly modified by
splitting equivalence classes until a robust stutter bisimulation is found.

\begin{definition}\label{def:splitter}
	Let $G=\tsystem$ be a transition system, and let \RR\ be an equivalence relation on~$S$. An \RR-superblock step formula $\psi \equiv P \anyuntil T$ with $P \in S/\RR$ is called a \emph{splitter} of \RR\ if $\emptyset \neq P \cap \ECS{\psi} \neq P$.
\end{definition}

Splitters are formulas that cannot be enforced from all states of an equivalence class. 
The existence of a splitter means that a relation
cannot be a robust stutter bisimulation as it would violate
condition~\ref{it:rsb:equivalent} in
\defn~\ref{def:robust:stutter:bisimulation}.

\def\PropSplitter{%
	Let $G=\tsystem$ be a transition system, and let~\RR\ be an equivalence
	relation on~$S$. Then~\RR\ is a robust stutter bisimulation on~$G$ if and
	only if it satisfies condition~\ref{it:rsb:label} in
	\defn~\ref{def:robust:stutter:bisimulation} and there does not exist any
	splitter of~\RR.}

\begin{proposition}\label{prop:splitter}
	\PropSplitter
\end{proposition}

\appref{\App~\ref{app:algorithm} contains proofs of all propositions in this
section.}
\Propn~\ref{prop:splitter} gives a stopping condition; once no
splitters can be found, an equivalence relation is a robust stutter
bisimulation. If there is a splitter, the relation can be refined such that it
no longer has this splitter.

\begin{definition}
	\label{def:refine}
	Let $G=\tsystem$ be a transition system, let \RR\ be an equivalence relation on~$S$, and let $\psi \equiv P \anyuntil T$ be a splitter of \RR. The \emph{refinement} of \RR\ by~$\psi$ is
	\begin{equation}
		\label{eq:refine}
		\Refine(\RR, \psi) = \RR \setminus (D \cup D^{-1}) \ ,
	\end{equation}
	where $D = (P \cap \ECS{\psi}) \times (P \setminus \ECS{\psi})$ and $D \cup D^{-1}$ is the symmetric closure of~$D$.
\end{definition}

The refinement of~\RR\ by~$\psi$ is a new equivalence relation where the
equivalence class~$P$ is replaced by two smaller equivalence classes $P
\cap \ECS\psi$ and $P \setminus \ECS\psi$, thus ensuring that the
formula~$\psi$ can either be enforced or not from two equivalent
states. This refinement always results in a finer relation, $\Refine(\RR,
\psi) \subseteq \RR$. Moreover the following proposition shows that, if a
relation is at least as coarse as the coarsest robust stutter bisimulation~\inbisim,
then its refinement by a splitter retains this property.

\def\PropRefineCoarseness{%
	Let $G=\tsystem$ be a transition system, let \RR\ be an equivalence
	relation on~$S$, and let $\psi \equiv P \anyuntil T$ be a splitter
	of~\RR. If $\inbisim \subseteq \RR$ then also $\inbisim \subseteq
	\Refine(\RR,\psi)$.}

\begin{proposition}\label{prop:refine:coarseness}
	\PropRefineCoarseness
\end{proposition}

Combining the results of \props\ \ref{prop:splitter}
and~\ref{prop:refine:coarseness} leads to
Algorithm~\ref{alg:QuotientPartition}, which computes the coarsest robust
stutter bisimulation for a transition system~$G$. Starting with an initial
equivalence relation~$\RR^0$ defined based on the state labels, which
clearly is at least as coarse as~\inbisim, the algorithm checks for the existence of
splitters. If there are none, the relation is a robust stutter bisimulation
by \Propn~\ref{prop:splitter}, and is returned. Otherwise it is refined
based on \defn~\ref{def:refine}, resulting in a new equivalence relation
still at least as coarse as~\inbisim\ according to \Propn~\ref{prop:refine:coarseness}. If the loop terminates, the result is a robust stutter bisimulation at least as coarse as~\inbisim, which must be equal to~\inbisim.

\IncMargin{1.1em}
\begin{algorithm}[tbp]
	\KwIn{$G = \tsystem$}
	\KwOut{Coarsest robust stutter bisimulation \inbisim}
	\BlankLine
	$\RR^0 \gets \{\, (s,t) \in S \times S \mid L(s) = L(t) \,\}$\;
	$i \gets 0$\;
	\While{\rm there exists a splitter $\psi$ of $\RR^i$\label{l:loop}}{
		$\RR^{i+1} \gets \Refine(\RR^i,\psi)$\;
		$i \gets i+1$\;
	}
	\Return{$\RR^i$}\;
	\caption{Coarsest robust stutter bisimulation\kern-1em}
	\label{alg:QuotientPartition}
\end{algorithm}

\def\PropAlgorithm{%
	If Algorithm~\ref{alg:QuotientPartition} terminates on an input transition
	system~$G$, then the result is the coarsest robust stutter
	bisimulation \inbisim\ on~$G$.}

\begin{proposition}\label{prop:algorithm}
	\PropAlgorithm
\end{proposition}

The loop entry condition in Algorithm~\ref{alg:QuotientPartition} requires
a search for splitters, which can be performed by enumerating all
$\RR^i$-superblock step formulas $P \anyuntil T$.
Termination of Algorithm~\ref{alg:QuotientPartition} is only guaranteed for
finite-state transition systems. For infinite-state systems, the loop may
find new splitters indefinitely, or the computation of $\ECS\psi$ to check
whether a formula~$\psi$ is a splitter may fail to terminate.
If Algorithm~\ref{alg:QuotientPartition} terminates, then \inbisim\ has a
finite number of equivalence classes, and the quotient $G/\inbisim$ can be
constructed to solve the \ref{problem:abstraction}.

\section{Constructing a concrete controller}
\label{sec:concreteController}

Given two robust stutter bisimilar transition systems and a controller for
one of them, this section demonstrates the construction of a controller for
the other transition system that enforces the same \LTLnn\ formulas. This
serves two purposes. Firstly, it solves the \ref{problem:construction}
by showing how to construct a controller for a concrete system from a
controller synthesised for the abstract system. Secondly, it establishes a proof
of the claimed property of robust stutter bisimulations that the same
\LTLnn\ formulas can be enforced on robust stutter bisimilar transition
systems.

Given two robust stutter bisimilar transition systems $G$ and~$\Tilde{G}$
and a controller~$\Tilde{C}$ enforcing an \LTLnn\ formula~$\varphi$
on~$\Tilde{G}$, the goal is to construct a controller~$C$ that
enforces~$\varphi$ on~$G$. For the construction of a concrete controller,
$\Tilde{G}$ can be assumed to be a quotient $\Tilde{G} = G/\RR$. In the
more general proof of equivalence, $G$ and~$\Tilde{G}$ can be arbitrary
robust stutter bisimilar transition systems.

As $G$ and~$\Tilde{G}$ are robust stutter bisimilar, there exists a
robust stutter bisimulation between them, which by
\defn~\ref{def:robust:stutter:between} is a robust stutter bisimulation on
the union transition system $G \cup \Tilde{G}$. Then $\Tilde{C}$ can be
understood as a controller that enforces~$\varphi$ from all
states~$\tilde{s}$ of $G \cup \Tilde{G}$ that are initial states
of~$\Tilde{G}$, and $C$ must enforce $\varphi$ from equivalent states~$s$
that are initial states of~$G$. Therefore, it is enough to consider a
robust stutter bisimulation~\RR\ on a single transition system $G =
\tsystem$, which may be a union transition system, and two equivalent
states $(s, \tilde{s}) \in \RR$. Given a controller $\Tilde{C}$ for~$G$
that enforces $\varphi$ from~$\tilde{s}$,
the goal now is to construct a controller $C$ for~$G$ that enforces
$\varphi$ from~$s$.

When the controller~$C$ observes a path fragment $\rho \in S^+$ with $s \prefix \rho$, it must decide which control action to take next. This decision will be guided by the actions taken by~$\Tilde{C}$ in the following way. First the path fragment~$\rho$ seen by~$C$ is mapped by $M\colon S^+ \to S^+$ to a path fragment~$\tilde\rho$ for which $\tilde{s} \prefix \tilde\rho$ and which is permitted by~$\Tilde{C}$. Then the control action taken by~$\Tilde{C}$ when observing~$\tilde\rho$ is used to inform which control action $C$ takes.

$M$ maps a path fragment $\rho \in S^+$ with $s \prefix \rho$ to a nonempty stutter equivalent path fragment~$\tilde\rho \in S^+$ with $\tilde{s} \prefix \tilde\rho$, which is defined recursively.
Assume that some path fragment  $\rho = \rho' u$ with $\rho' \in S^*$ and $u \in S$ has been mapped to $\tilde\rho = \tilde\rho' \tilde{u}$ with $\tilde\rho' \in S^*$ and $\tilde{u} \in S$. If~$G$ under control of $C$ now enters a state $v \in S$ in a new equivalence class, i.e., $(u,v) \notin \RR$, then the construction of~$M$ considers all continuations of $\tilde\rho = \tilde\rho' \tilde{u}$ permitted by~$\Tilde{C}$ that remain in the equivalence class of~$\tilde{u}$ until reaching a state~$\tilde{v}$ equivalent to~$v$, namely
\begin{equation}\label{eq:F}
  F(\tilde{\rho}'\tilde{u}, v) = \LongSet{10em}{$\tilde{\rho}' \tilde{u} \tilde{\tau} \tilde{v} \in \Frags^*(\Tilde{C}, G) \mid \tilde{\tau} \in [\tilde{u}]_\RR^*$ and $(v,\tilde{v}) \in \RR$}.
\end{equation}
Then $M$ is defined by choosing one of these paths. The full recursive
definition of $M(\rho)$ for $\rho \in S^+$ with $s \prefix \rho$ is as
follows:
\begin{align}
  M(s) & = \tilde{s} \ ; \label{eq:M:init} 
  \\
  M(\rho' u v) & = M(\rho'u) && \text{if}\ (u,v) \in \RR\ ; \label{eq:M:stay} 
  \\
  M(\rho' u v) & = \tau \in F(M(\rho'u),v) && \text{if}\ (u,v) \notin \RR\ ;
  \label{eq:M:step}
\end{align}
where $\tau$ is an arbitrary but fixed choice. If $M(\rho'u)$ is undefined
or $F(M(\rho'u), v) = \emptyset$ in \eqref{eq:M:stay} or~\eqref{eq:M:step},
then $M(\rho'u v)$ is undefined. The following construction of~$C$ ensures that this never happens when $\rho$ is permitted by~$C$, so $M(\rho)$ is defined for all $\rho \in \Frags^*(C,G)$ with $s \prefix \rho$.
Also, if $M(\rho)$ is defined, then $M(\rho) \in \Frags^*(\Tilde{C}, G)$, $\tilde{s} \sqsubseteq M(\rho)$, and $M(\rho)$ and~$\rho$ are stutter equivalent, or, more precisely, $M(\rho)$ and~$\rho$ visit the same equivalence classes apart from possible repetitions. Furthermore, $M$ is prefix-preserving, i.e., $\rho' \sqsubseteq \rho$ implies $M(\rho') \sqsubseteq M(\rho)$.

\begin{example}\label{ex:controller:construction:M}
Consider again the transition system~$G$ and the robust stutter
bisimulation~\RR\ in \fig~\ref{fig:quotient:system:concrete}, and its
abstraction~$\Tilde{G} = G/\RR$ from \Examp~\ref{ex:quotient:system}. The
relation \RR\ on~$S$ can be extended to a relation~$\hat\RR$ on $S \cup
(S/\RR)$ by including the abstract states in their equivalence classes,
e.g., $[\Aq_1]_{\hat\RR} = \{ a_1 \bcom a_2 \bcom \Aq_1 \}$. It can be
shown that this extension~$\hat\RR$ is a robust stutter bisimulation on $G
\cup \Tilde{G}$. The initial state $s = d_1$ of $G \cup \Tilde{G}$ is
equivalent to $\tilde{s} = \Dq$.
Then $M(d_1) = \Dq$ according to~\eqref{eq:M:init}.
For the path fragment~$d_1d_2$, the result of the mapping becomes $M(d_1d_2) = M(d_1) = \Dq$ by~\eqref{eq:M:stay}.
Let~$\Tilde{C}$ be defined as $\Tilde{C}(\tilde{\rho} \tilde{v}) = \{ \Dq \weakuntil \Cq \}$ when $\tilde{v} = \Dq$, and $\Tilde{C}(\tilde{\rho} \tilde{v}) = \Sigma \cup \Sigma_\RR$ otherwise. It follows that $\Dq^k\Cq \in \Frags^*(\Tilde{C}, G \cup \Tilde{G})$ for all $k \geq 1$.
The path fragment~$d_1d_2c_1$ is then mapped to a fixed element of $F(M(d_1d_2), c_1) = F(\Dq, c_1) = \{ \Dq\Cq\bcom \Dq\Dq\Cq\bcom \Dq\Dq\Dq\Cq, \ldots \}$ by~\eqref{eq:M:step}, one possibility being $M(d_1d_2c_1) = \Dq\Cq$.\qed
\end{example}

With the mapping~$M$ defined, the control action of~$C$ when observing
$\rho = \rho'u$ can now be based on the control action of the given
controller~$\Tilde{C}$ when observing $M(\rho)$. Assume $M(\rho) =
M(\rho'u) = \tilde{\rho}'\tilde{u} = \tilde\rho$. Since $\Tilde{C}$ is
\dlfree, it enforces at least one \RR-\hskip0pt superblock step
formula~$\psi$ from~$\tilde{u}$ after seeing $\tilde{\rho}$. As the
construction of~$M$ ensures that $u$ and~$\tilde{u}$ are robust stutter
bisimilar, the same formula~$\psi$ can be enforced from~$u$. Therefore, $C$
can be constructed such that its control decision when seeing~$\rho$
follows the control decision of a \memoryless\ controller~$\overlinit{C}$
that enforces~$\psi$ from~$u$. There may be several such controllers
enforcing~$\psi$, in which case an arbitrary but fixed choice is made for a given
path~$\tilde{\rho}$.

Formally, $\psi$ is chosen such that its target set is the \emph{superblock reached} by~$\Tilde{C}$ after~$\tilde{\rho}$, defined by
\begin{align}
  \label{eq:TAfter:tilde}
  \kern-.5em
  T\langle \Tilde{C}, \tilde{\rho}'\tilde{u} \rangle =
  \bigcup \, \LongSet{10em}{$[\tilde{v}]_\RR \mid
    \tilde{\rho}'\tilde{u}\tilde{\tau}\tilde{v} \in \Frags^*(\Tilde{C},
    G)$ \\ where $\tilde{\tau} \in [\tilde{u}]_\RR^*$ and
    $(\tilde{u}, \tilde{v}) \notin \RR$}.
\end{align}
The formula~$\psi$ is chosen to be the \emph{\RR-superblock step
  formula enforced} by~$\Tilde{C}$ after~$\tilde{\rho}$, defined as
\begin{align}
  \label{eq:psiAfter:tilde}
  \Psi\langle \Tilde{C}, \tilde{\rho}'\tilde{u} \rangle \equiv 
  \begin{cases}
    [\tilde{u}]_\RR \weakuntil
    T\langle \Tilde{C}, \tilde{\rho}'\tilde{u} \rangle, &  \\
    \multicolumn{2}{r@{}}{\text{if}\ \Frags^\omega(\Tilde{C}, G) \cap
                        \tilde\rho [\tilde{u}]_\RR^\omega \neq \emptyset\ ;} \\
    [\tilde{u}]_\RR \until T\langle \Tilde{C}, \tilde{\rho}'\tilde{u} \rangle,
    & \quad\text{otherwise .} \\
  \end{cases}
\end{align}
The superblock reached by~$\Tilde{C}$ after~$\tilde{\rho}$ contains exactly the states in equivalence classes that can be entered directly after the equivalence class of~$\tilde{u}$. The \RR-superblock step formula enforced by~$\Tilde{C}$ after~$\tilde{\rho}$ is a condition in the form of a weak or strong until formula that describes how the controller behaves within the equivalence class of~$\tilde{u}$ reached by~$\tilde{\rho}$. If the controller permits a path fragment that remains in the equivalence class of~$\tilde{u}$ forever, then the enforced \RR-superblock step formula is a weak until formula. Otherwise all permitted path fragments eventually leave the equivalence class of~$\tilde{u}$, and the enforced formula is instead a strong until formula. In both cases, the source set is the equivalence class of the current state~$\tilde{u}$, and the target set is the superblock reached after~$\tilde{\rho}$, which contains the states that can be entered after the equivalence class of~$\tilde{u}$.

\begin{example}\label{ex:controller:construction:Psi}
Consider the controller~$\Tilde{C}$ and the path fragment $\tilde{s} = \Dq$ from \Examp~\ref{ex:controller:construction:M}. Under control of~$\Tilde{C}$, a path
fragment starting in~$\Dq$ either stays in~$\Dq$ forever or transits to~$\Cq$
eventually. Hence, the superblock reached
by~$\Tilde{C}$ after~$\Dq$ is $T\langle \Tilde{C}, \Dq \rangle = [\Cq]_{\hat\RR}$, and the $\hat\RR$-superblock step formula enforced by~$\Tilde{C}$
after~$\Dq$ is $\Psi\langle \Tilde{C}, \Dq \rangle \equiv [\Dq]_{\hat\RR}
\weakuntil [\Cq]_{\hat\RR}$. Note that the first case in~\eqref{eq:psiAfter:tilde} is chosen because $\Dq^\omega \in \Frags^\omega(\Tilde{C}, G \cup \Tilde{G})$ and $\Dq^\omega \in [\Dq]_{\hat\RR}^\omega$. \qed
\end{example}

\appref{It is shown in \app~\ref{app:LTLnnPreservation}}
\shorttext{It can be shown}
that $\psi \equiv \Psi
\langle \Tilde{C}, \tilde{\rho} \rangle$ is indeed enforceable. More
precisely, there exists a controller enforcing~$\psi$ from every state
equivalent to the end state $\tilde{u}$ of~$\tilde\rho$. By choosing one
such controller~$\overlinit{C}\langle \psi \rangle$ and using it while staying
in the equivalence class of~$\tilde{u}$, the control decision of~$C$ can
finally be defined.

\begin{definition}\label{def:concrete:emulator}
  Let $G$ be a transition system, and let $\psi \equiv P \anyuntil T$ be a
  stutter step formula for~$G$. Then define $\overlinit{C}\langle \psi
  \rangle$ to be a \memoryless\ controller such that $\langle
  \overlinit{C}\langle \psi \rangle /G, u \rangle \vDash \psi$ for all $u
  \in P$, if such controller exists.
\end{definition}

\begin{definition}\label{def:concrete:controller}
Let \RR\ be a robust stutter bisimulation on a transition system
$G$, let $\Tilde{C}$ be a controller for~$G$, and
let $M$ and~$\Psi$ be defined as above. The \emph{concrete
  controller}~$C\colon S^+ \to 2^\Sigma$ for~$G$ based on~$\Tilde{C}$ is
defined by
\begin{equation*}
  C(\rho'u) =
  \begin{cases}
    \overlinit{C} \langle \Psi \langle \Tilde{C},
    M(\rho'u) \rangle \rangle (u), \kern-3pt &
              \text{if $M(\rho'u)$ is defined}; \\
    \Sigma, & \text{otherwise}.
  \end{cases}
\end{equation*}
\end{definition}

\begin{example}\label{eq:controller:construction}
Let $\Tilde{C}$ be defined as in \examp~\ref{ex:controller:construction:M}.
Assume a controller~$C$ shall be constructed based on~$\Tilde{C}$ for path
fragments starting in the state~$d_1$. Then, for example, $C(d_1)$ is
obtained following \defn~\ref{def:concrete:controller} and
\examps\ \ref{ex:controller:construction:M}
and~\ref{ex:controller:construction:Psi} as $C(d_1) = \overlinit{C}\langle
\Psi \langle \Tilde{C}, M(d_1) \rangle \rangle (d_1) = \overlinit{C}\langle
\Psi \langle \Tilde{C}, \Dq \rangle \rangle (d_1) = \overlinit{C}\langle
[\Dq]_{\hat\RR} \weakuntil [\Cq]_{\hat\RR} \rangle (d_1)$. The only
\memoryless\ controller enforcing $\Psi \langle \Tilde{C}, \Dq\rangle
\equiv [\Dq]_{\hat\RR} \weakuntil [\Cq]_{\hat\RR}$ from $d_1$ and~$d_2$ is
$\overlinit{C}(d_i) = \{ \sigma_1 \}$. Hence, $C(d_1) = \{ \sigma_1 \}$.
Similarly, as $M(d_1d_2) = M(d_1) = \Dq$ from
\examp~\ref{ex:controller:construction:M}, it follows that $C(d_1d_2) =
\overlinit{C}\langle [\Dq]_{\hat\RR} \weakuntil [\Cq]_{\hat\RR} \rangle
(d_2) = \{\sigma_1\}$. \qed
\end{example}

\Thm~\ref{thm:LTLnnPreservation} confirms that the same \LTLnn\ formulas
can be enforced from robust stutter bisimilar states.
\Crl~\ref{cor:LTLnn:preservation} lifts this result to robust stutter
bisimilar transition systems, showing the preservation of 
\LTLnn\ synthesis results under robust stutter bisimulation.

\def\ThmLTLnnPreservation{%
  Let $G$ be a transition systems, let \RR\ be a
  robust stutter bisimulation on $G$, and let $\varphi$ be an
  \LTLnn\ formula. For all $(\tilde{s}, s) \in \RR$ such that $\tilde{s} \in
  \EC[G]{\varphi}$ it also holds that $s \in
  \EC[G]{\varphi}$.}

\begin{theorem}
  \label{thm:LTLnnPreservation}
  \ThmLTLnnPreservation
\end{theorem}

\def\CorLTLnnPreservation{%
  Let $G$ and $\Tilde{G}$ be transition systems, let \RR\ be
  a robust stutter bisimulation between them, and let $\varphi$ be an
  \LTLnn\ formula. If there exists a \dlfree controller~$\Tilde{C}$
  for~$\Tilde{G}$ such that $\Tilde{C}/\Tilde{G} \models \varphi$, then
  there exists a \dlfree controller~$C$ for~$G$ such that $C/G
  \models \varphi$.}

\begin{corollary}
  \label{cor:LTLnn:preservation}
  \CorLTLnnPreservation
\end{corollary}

\appref{The full proof of these results, based on the above controller
construction, is given in \app~\ref{app:LTLnnPreservation}.}
If the
concrete controller~$C$ is constructed in this way from a
controller~$\Tilde{C}$ for an abstract transition system~$\Tilde{G}$ that
is robust stutter bisimilar to a concrete transition system~$G$, then it is
a solution to the \ref{problem:construction}.

\section{Robot navigation example}
\label{sec:example}

This section applies the abstraction and controller construction to a robot
navigation problem inspired by \citet{MohMalWinLafOza:21}. In this example,
a robot navigates a two-dimensional space while avoiding the \emph{Obstacle} 
shown in \fig~\ref{fig:example:concrete}. The main difference to the
previous work is the addition of disturbance and an area that is too narrow
for the robot to pass due to the disturbance.

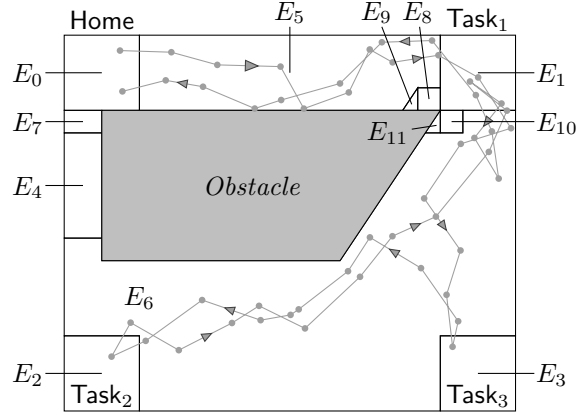
\begin{figure}
\centering
\scalebox{1}{
\begin{tikzpicture}
\draw (5.0000, -0.0000) -- (5.1666, -0.0000) -- (5.3000, -0.0000) -- (5.6665, -0.0000) -- (6.0000, -0.0000) -- (6.0000, -0.7000) -- (6.0000, -1.0000) -- (5.3000, -1.0000) -- (5.0000, -1.0000) -- (5.0000, -0.7000) -- (5.0000, -0.2499) -- cycle;
\draw (-0.0000, -5.0000) -- (-0.0000, -4.0000) -- (0.5000, -4.0000) -- (1.0000, -4.0000) -- (1.0000, -5.0000) -- (0.5000, -5.0000) -- cycle;
\draw (6.0000, -5.0000) -- (6.0000, -4.0000) -- (5.3000, -4.0000) -- (5.0000, -4.0000) -- (5.0000, -5.0000) -- (5.3000, -5.0000) -- cycle;
\draw (-0.0000, 0.0000) -- (0.5000, -0.0000) -- (1.0000, -0.0000) -- (1.0000, -0.7000) -- (1.0000, -1.0000) -- (0.5000, -1.0000) -- (-0.0000, -1.0000) -- (-0.0000, -0.7000) -- cycle;
\draw (5.0000, -1.0000) -- (4.7000, -1.0000) -- (4.7000, -0.7000) -- (5.0000, -0.7000) -- cycle;
\draw (-0.0000, -1.0000) -- (-0.0000, -1.3000) -- (0.5000, -1.3000) -- (0.5000, -1.0000) -- cycle;
\draw (1.0000, -0.0000) -- (4.7000, -0.0000) -- (5.0000, -0.0000) -- (5.0000, -0.2499) -- (5.0000, -0.7000) -- (4.7000, -0.7000) -- (4.5001, -1.0000) -- (1.0000, -1.0000) -- (1.0000, -0.7000) -- cycle;
\draw (6.0000, -1.0000) -- (6.0000, -1.3000) -- (6.0000, -2.7000) -- (6.0000, -3.0000) -- (6.0000, -4.0000) -- (5.3000, -4.0000) -- (5.0000, -4.0000) -- (5.0000, -5.0000) -- (4.7000, -5.0000) -- (2.3340, -5.0000) -- (1.8341, -5.0000) -- (1.0000, -5.0000) -- (1.0000, -4.0000) -- (0.5000, -4.0000) -- (-0.0000, -4.0000) -- (-0.0000, -3.0000) -- (-0.0000, -2.7000) -- (0.5000, -2.7000) -- (0.5000, -3.0000) -- (1.0000, -3.0000) -- (3.1671, -3.0000) -- (3.6670, -3.0000) -- (3.8670, -2.7000) -- (4.7000, -1.4501) -- (4.8001, -1.3000) -- (5.0000, -1.3000) -- (5.3000, -1.3000) -- (5.3000, -1.0000) -- cycle;
\draw (4.7000, -1.0000) -- (4.5001, -1.0000) -- (4.7000, -0.7000) -- cycle;
\draw (5.0000, -1.3000) -- (5.0000, -1.0000) -- (4.8001, -1.3000) -- cycle;
\draw (-0.0000, -1.3000) -- (-0.0000, -2.7000) -- (0.5000, -2.7000) -- (0.5000, -1.3000) -- cycle;
\draw (5.0000, -1.0000) -- (5.3000, -1.0000) -- (5.3000, -1.3000) -- (5.0000, -1.3000) -- cycle;
\draw [fill=gray!50] (5.0000, -1.0000) -- (4.7000, -1.0000) -- (4.5001, -1.0000) -- (1.0000, -1.0000) -- (0.5000, -1.0000) -- (0.5000, -1.3000) -- (0.5000, -2.7000) -- (0.5000, -3.0000) -- (1.0000, -3.0000) -- (3.1671, -3.0000) -- (3.6670, -3.0000) -- (3.8670, -2.7000) -- (4.7000, -1.4501) -- (4.8001, -1.3000) -- cycle;

\draw [draw=gray!75] (0.7418, -0.2048) node [circle, fill=gray!75, minimum size=2.5pt, inner sep=0] {{}} -- (1.4245, -0.2267) node [circle, fill=gray!75, minimum size=2.5pt, inner sep=0] {{}} -- (2.0354, -0.3951) node [circle, fill=gray!75, minimum size=2.5pt, inner sep=0] {{}} -- (2.8757, -0.4181) node [midway, isosceles triangle, fill=gray!75, draw=black!75, inner sep=0, minimum size=4pt, rotate=-1.5690709913605931] {} node [circle, fill=gray!75, minimum size=2.5pt, inner sep=0] {{}} -- (3.1861, -0.9774) node [circle, fill=gray!75, minimum size=2.5pt, inner sep=0] {{}} -- (3.7834, -0.7618) node [circle, fill=gray!75, minimum size=2.5pt, inner sep=0] {{}} -- (4.0618, -0.1900) node [circle, fill=gray!75, minimum size=2.5pt, inner sep=0] {{}} -- (4.3660, -0.3617) node [circle, fill=gray!75, minimum size=2.5pt, inner sep=0] {{}} -- (4.9750, -0.2667) node [midway, isosceles triangle, fill=gray!75, draw=black!75, inner sep=0, minimum size=3pt, rotate=8.86314793467138] {} node [circle, fill=gray!75, minimum size=2.5pt, inner sep=0] {{}} -- (5.5031, -0.5617) node [circle, fill=gray!75, minimum size=2.5pt, inner sep=0] {{}} -- (5.9400, -1.2379) node [circle, fill=gray!75, minimum size=2.5pt, inner sep=0] {{}} -- (5.2723, -1.4458) node [circle, fill=gray!75, minimum size=2.5pt, inner sep=0] {{}} -- (4.7854, -2.1717) node [circle, fill=gray!75, minimum size=2.5pt, inner sep=0] {{}} -- (5.2687, -2.8642) node [midway, isosceles triangle, fill=gray!75, draw=black!75, inner sep=0, minimum size=3pt, rotate=-55.09046832523555] {} node [circle, fill=gray!75, minimum size=2.5pt, inner sep=0] {{}} -- (5.0830, -3.4541) node [circle, fill=gray!75, minimum size=2.5pt, inner sep=0] {{}} -- (5.1629, -4.1456) node [circle, fill=gray!75, minimum size=2.5pt, inner sep=0] {{}} -- (5.2377, -3.8063) node [circle, fill=gray!75, minimum size=2.5pt, inner sep=0] {{}} -- (4.7434, -3.0994) node [circle, fill=gray!75, minimum size=2.5pt, inner sep=0] {{}} -- (4.0643, -2.6900) node [midway, isosceles triangle, fill=gray!75, draw=black!75, inner sep=0, minimum size=3pt, rotate=148.91696455273447] {} node [circle, fill=gray!75, minimum size=2.5pt, inner sep=0] {{}} -- (3.7731, -3.1404) node [circle, fill=gray!75, minimum size=2.5pt, inner sep=0] {{}} -- (3.1176, -3.6478) node [circle, fill=gray!75, minimum size=2.5pt, inner sep=0] {{}} -- (3.0094, -3.7214) node [circle, fill=gray!75, minimum size=2.5pt, inner sep=0] {{}} -- (2.6111, -3.7919) node [circle, fill=gray!75, minimum size=2.5pt, inner sep=0] {{}} -- (1.8336, -3.5241) node [midway, isosceles triangle, fill=gray!75, draw=black!75, inner sep=0, minimum size=3pt, rotate=160.99524067021983] {} node [circle, fill=gray!75, minimum size=2.5pt, inner sep=0] {{}} -- (1.0833, -4.0670) node [circle, fill=gray!75, minimum size=2.5pt, inner sep=0] {{}} -- (0.6294, -4.2768) node [circle, fill=gray!75, minimum size=2.5pt, inner sep=0] {{}} -- (0.8829, -3.8247) node [circle, fill=gray!75, minimum size=2.5pt, inner sep=0] {{}} -- (1.5234, -4.1910) node [circle, fill=gray!75, minimum size=2.5pt, inner sep=0] {{}} -- (2.2304, -3.8294) node [midway, isosceles triangle, fill=gray!75, draw=black!75, inner sep=0, minimum size=3pt, rotate=27.086027906332827] {} node [circle, fill=gray!75, minimum size=2.5pt, inner sep=0] {{}} -- (2.5904, -3.6011) node [circle, fill=gray!75, minimum size=2.5pt, inner sep=0] {{}} -- (3.2011, -3.9014) node [circle, fill=gray!75, minimum size=2.5pt, inner sep=0] {{}} -- (3.9328, -3.2147) node [circle, fill=gray!75, minimum size=2.5pt, inner sep=0] {{}} -- (4.3635, -2.6745) node [circle, fill=gray!75, minimum size=2.5pt, inner sep=0] {{}} -- (4.9417, -2.4138) node [midway, isosceles triangle, fill=gray!75, draw=black!75, inner sep=0, minimum size=3pt, rotate=24.265972571905994] {} node [circle, fill=gray!75, minimum size=2.5pt, inner sep=0] {{}} -- (5.6448, -1.5530) node [circle, fill=gray!75, minimum size=2.5pt, inner sep=0] {{}} -- (5.8862, -1.0041) node [circle, fill=gray!75, minimum size=2.5pt, inner sep=0] {{}} -- (5.3945, -0.6167) node [circle, fill=gray!75, minimum size=2.5pt, inner sep=0] {{}} -- (5.8136, -0.9090) node [circle, fill=gray!75, minimum size=2.5pt, inner sep=0] {{}} -- (5.4533, -1.3669) node [midway, isosceles triangle, fill=gray!75, draw=black!75, inner sep=0, minimum size=3pt, rotate=-128.1922711594843] {} node [circle, fill=gray!75, minimum size=2.5pt, inner sep=0] {{}} -- (5.7764, -1.9079) node [circle, fill=gray!75, minimum size=2.5pt, inner sep=0] {{}} -- (5.6851, -1.2799) node [circle, fill=gray!75, minimum size=2.5pt, inner sep=0] {{}} -- (5.4726, -0.7973) node [circle, fill=gray!75, minimum size=2.5pt, inner sep=0] {{}} -- (4.8900, -0.0712) node [circle, fill=gray!75, minimum size=2.5pt, inner sep=0] {{}} -- (4.2233, -0.0942) node [midway, isosceles triangle, fill=gray!75, draw=black!75, inner sep=0, minimum size=3pt, rotate=-178.02194158568872] {} node [circle, fill=gray!75, minimum size=2.5pt, inner sep=0] {{}} -- (3.6668, -0.6425) node [circle, fill=gray!75, minimum size=2.5pt, inner sep=0] {{}} -- (2.8842, -0.8627) node [circle, fill=gray!75, minimum size=2.5pt, inner sep=0] {{}} -- (2.5332, -0.9761) node [circle, fill=gray!75, minimum size=2.5pt, inner sep=0] {{}} -- (1.8456, -0.6858) node [circle, fill=gray!75, minimum size=2.5pt, inner sep=0] {{}} -- (1.3497, -0.6009) node [midway, isosceles triangle, fill=gray!75, draw=black!75, inner sep=0, minimum size=3pt, rotate=170.29272516451533] {} node [circle, fill=gray!75, minimum size=2.5pt, inner sep=0] {{}} -- (0.7738, -0.7479) node [circle, fill=gray!75, minimum size=2.5pt, inner sep=0] {{}};

\node at (2.5, -2) {\emph{Obstacle}};

\node at (0.5, 0.2) {\home};
\node at (-0.25, -0.5) [left, inner sep=0.5pt, text depth=0] (ea) {\enuma};
\draw (0.5, -0.5) -- (ea);
\node at (3, 0.3) [inner sep=1pt] (eb) {\enumb};
\draw (3, -0.5) -- (eb);
\node at (5.5, 0.2) {\taskc};
\node at (6.25, -0.5) [right, inner sep=0.5pt, text depth=0] (ec) {\enumc};
\draw (5.5, -0.5) -- (ec);

\node at (4.15, 0.3) [inner sep=1pt] (ed) {\enumd};
\draw (4.62, -0.92) -- (ed);
\node at (4.7, 0.3) [inner sep=1pt] (ee) {\enume};
\draw (4.85, -0.85) -- (ee);

\node at (-0.25, -1.15) [left, inner sep=0.5pt, text depth=0] (ef) {\enumf};
\draw (0.25, -1.15) -- (ef);
\node at (4.3, -1.3) [inner sep=1pt] (eg) {\enumg};
\draw (4.95, -1.2) -- (eg);
\node at (6.25, -1.15)  [right, inner sep=0.5pt, text depth=0] (eh) {\enumh};
\draw (5.15, -1.15) -- (eh);
\node at (1, -3.5) {\enumi};

\node at (-0.25, -2) [left, inner sep=0.5pt, text depth=0] (ej) {\enumj};
\draw (0.25, -2) -- (ej);

\node at (0.5, -4.8) {\taskk};
\node at (-0.25, -4.5) [left, inner sep=0.5pt, text depth=0] (ek) {\enumk};
\draw (0.5, -4.5) -- (ek);
\node at (5.5, -4.8) {\taskl};
\node at (6.25, -4.5) [right, inner sep=0.5pt, text depth=0] (el) {\enuml};
\draw (5.5, -4.5) -- (el);
\end{tikzpicture}
} 
\caption{The equivalence classes of the coarsest \rdsb for the robot navigation example.}
\label{fig:example:concrete}
\end{figure}

The robot movement is described by the following discrete-time linear system with disturbances
\begin{equation}
x(k+1) = x(k) + u(k) + w(k) \ , \label{eq:example:delta}
\end{equation}
where the state is $x \in \mathcal{X} = ([0, 6] \times [0, 5]) \setminus \mathit{Obstacle}$, the control input is $u \in \mathcal{U} = [-0.6, 0.6]^2$, and the disturbance is $w \in \mathcal{W} = [-0.3, 0.3]^2$.
The dynamics described by~\eqref{eq:example:delta} is used to construct a
transition system~$G$ as in \sect~\ref{sec:transsys}. The set of
initial states is the top left corner, $S\init = \enuma = [0, 1] \times [4, 5]$. The
set of propositions is $\AP = \{ \home, \taskc, \taskk, \taskl \}$, and the
state labelling function~$L$ assigns these propositions to the corner sets, see \fig~\ref{fig:example:concrete}. Thus, the initial equivalence
relation~$\RR^0$ for \alg~\ref{alg:QuotientPartition} has five
equivalence classes: the four corners and the rest of the state space.

When \alg~\ref{alg:QuotientPartition} is applied to the transition
system~$G$, the resulting \rdsb\ \RR\ consists of twelve equivalence
classes $E_0,\ldots,E_{11}$ shown in \fig~\ref{fig:example:concrete}.
The four corners are not split and remain as regions
\enuma, \enumc, \enumk, and~\enuml. It is possible to keep the robot
indefinitely in the corners, because the disturbance cannot push the
robot out of these regions from the central positions and the control input
has higher magnitude than any possible disturbance. For example
$\enuma \weakuntil \emptyset$ is enforceable from all states in~\enuma,
i.e., $\enuma \subseteq \ECS{\enuma \weakuntil \emptyset}$ and thus 
$\enuma \weakuntil \emptyset \in \Sigma_\RR$.

The region between the corners is split into eight equivalence classes.
\enumj\ is split off from the other regions because it is so narrow that, independently of the control input, the disturbance may push the robot into the \emph{Obstacle} from any position within it,  and there is no \RR-superblock step formula from \enumj\ in $\Sigma_\RR$. From all other regions there exists some enforceable \RR-superblock step formula.
The two large equivalence classes \enumb\ and \enumi\ are dissimilar
because the robot can be forced from \enumb\ to~\enuma\ without
visiting another region, but this cannot be enforced
from~\enumi.
That is, $\enumb \until \enuma \in \Sigma_\RR$ but
$\enumi \until \enuma \notin \Sigma_\RR$. Other superblock step formulas
that can be enforced from \enumb\ are $\enumb \until \enumc \in \Sigma_\RR$
and $\enumb \weakuntil \emptyset \in \Sigma_\RR$.
The remaining equivalence classes \enumf, \enume, \enumd, \enumh,
and~\enumg\ are so small in relation to the disturbance that the controller
cannot force the robot to enter or stay in these classes.
\enumf\ is different from \enumb, because from~\enumf\ the robot can only be forced
into~\enuma\ and not into~\enumc:  
$\enumf \until \enumc \notin \Sigma_\RR$ whereas $\enumb \until \enumc \in \Sigma_\RR$.
The equivalence classes \enume\ and~\enumd\ are split off from~\enumb\
because it is possible to force the robot directly into the superblock
$\enumi \cup \enumh$ from $\enume \cup \enumd$ but not
from~\enumb. That is, $\enume \until
(\enumi \cup \enumh) \in \Sigma_\RR$ but $\enumb \until
(\enumi \cup \enumh) \notin \Sigma_\RR$. Similar reasons
cause $\enumh \cup \enumg$ to be split off from~\enumi.
Then \enume\ and \enumd\ are separated since the robot can be forced
to~\enumc\ from \enume\ but not from~\enumd, i.e.,
$\enume \until \enumc \in \Sigma_\RR$ but
$\enumd \until \enumc \notin \Sigma_\RR$. The split of \enume\ and~\enumd\
propagates, causing a separation of \enumh\ and~\enumg, for example
$\enumg \until (\enumb \cup \enumd) \in \Sigma_\RR$ and $\enumh \until
(\enumb \cup \enume \cup \enumd) \in \Sigma_\RR$.

At this point, the partition is stable and no further split occurs. After
collecting all the superblock step formulas enforceable from each of the
twelve regions, an abstract transition system $\Tilde{G} = G/\RR$ is
constructed according to \defn~\ref{def:quotient:system}.
Then an abstract controller~$\Tilde{C}$ is synthesised for~$\Tilde{G}$ to enforce the \LTLnn formula
\begin{equation*}
	\globally\finally \home \land \globally\finally \taskc \land \globally\finally \taskk \land \globally\finally \taskl \ .
\end{equation*}
The abstract controller~$\Tilde{C}$ is used to construct a concrete
controller~$C$ according to \defn~\ref{def:concrete:controller}. \fig~\ref{fig:example:concrete} shows a path fragment where the robot completes
the three tasks and returns to the home region while exposed to the
disturbance, which causes an erratic behaviour.

\section{Conclusions}
\label{sec:conclusions}

This \whatsit\ proposes a method to synthesise controllers that enforce requirements specified in \LTLnn\ for cyber-physical systems subject to disturbances. This is done by constructing a finite-state abstraction of the system and then synthesising a controller for the abstract system. The \emph{\rdsb} relation is shown to characterise the relevant abstraction accurately, and it is shown how this relation can be used to construct a concrete controller. Although the main driver of this work is the ability to handle process noise, \rdsb is more general and can be used on any system that can be represented as transition systems.

For future usage an efficient implementation of \Alg~\ref{alg:QuotientPartition} is important. As given here, every superblock step formula is checked whether it is a splitter, though some implied formulas would not have to be checked. The number of checks may be further reduced by considering what other formulas did not lead to a split.

This \whatsit\ considers disturbances added to the state transitions, but for some applications there is considerable measurement noise that results in unobservable transitions. Future work might investigate how such disturbances can be incorporated in the abstraction.

{\hbadness=2500 \bibliography{malik_abrv,IEEEabrv,malik}}

\iffull
\appendix

\numberwithin{proposition}{section}
\numberwithin{definition}{section}
\numberwithin{remark}{section}

\section{Fixpoint characterisation of weak and strong until}
\label{app:fixpoints}

Robust stutter bisimulations are closely linked to the existence of
\memoryless\ controllers for stutter step formulas according to
\defn~\ref{def:robust:stutter:bisimulation}~\ref{it:rsb:equivalent}.
Therefore it is important for algorithms to compute the set
$\ECS[G]\psi$ of states from which a stutter step formula~$\psi$ can
be enforced. This appendix provides characterisations of $\ECS[G]{\psi}$
that do not require exhaustive searches for \memoryless\ controllers.

Such characterisations can be derived as fixpoints of a function on the
powerset of the state space. The first step towards the development of such
a function is to identify relevant sets of predecessors and successors of
states.

\begin{definition}
Let $G = \tsystem$ be a transition system, let $s \in S$ be a state, and
let $\sigma \in \Sigma$ be a transition label. The \emph{one-step
successor} sets $\Post(s, \sigma)$ and $\Post(s)$ are
\begin{align*}       
  \Post_G(s, \sigma)
  &= \{\, s' \in S \mid (s, \sigma, s') \in \delta \,\} \ ; \\
  \Post_G(s)
  &= \bigcup_{\sigma \in \Sigma} \Post(s, \sigma) \ .
\end{align*}
\end{definition}

\begin{definition}
	Let $G = \tsystem$ be a transition system. The set of \emph{one-step robust controllable predecessors} of a state set $X \subseteq S$ is
	\begin{equation*}
		\Pre_G(X) = \{\, s \in S\mid
		\exists \, \sigma \in \Sigma \ldotp
		\emptyset \neq \Post_G(s, \sigma) \subseteq X \,\}\ .
	\end{equation*}
\end{definition}

The subscript~$G$ will be omitted in the above notations if the transition
system is clear from the context. $\Pre(X)$ contains all states from which
it is possible to ensure, by selecting a transition label~$\sigma$, that
the next state is in~$X$. It is clear that $\Post$ and~$\Pre$ are
\emph\monotonic, e.g., if $X \subseteq Y \subseteq S$ then $\Pre(X)
\subseteq \Pre(Y)$.

\citet{BaiKat:08} describe fixpoint characterisations of weak and strong
until for model checking, which can be adapted to the synthesis problem
considered here. Given a transition system $G$ with state set~$S$ and a
stutter step formula $\psi \equiv P \anyuntil T$ with $P,T \subseteq S$,
the function $\Theta\colon 2^S \to 2^S$ is defined by
\begin{equation}
  \Theta(X) = T \cup (P \cap \Pre(X)) \ .
  \label{eq:Theta}
\end{equation}
The subsets of the state set $S$ can be ordered by set inclusion to form a
complete lattice~$\poset$. It is clear that $\Theta$ is
\monotonic\ on~$\poset$, because union with~$T$, intersection with~$P$, and
the $\Pre$ operator are all \monotonic.

As $\Theta$ is \monotonic\ on the complete lattice~$\poset$, it follows by the Knaster-Tarski theorem \citep{Tar:55} that there exists a unique \emph{least fixpoint} $\lfp \Theta$ and a unique \emph{greatest fixpoint} $\gfp \Theta$ of~$\Theta$. (A \emph{fixpoint} of~$\Theta$ is a set $\mathring{X} \subseteq S$ such that $\Theta(\mathring{X}) = \mathring{X}$.) Furthermore,
\begin{align}
\lfp \Theta &= \glb \{\, X \subseteq S \mid \Theta(X) \subseteq X \,\} \label{eq:tarski:lfp} \ ; \\
\gfp \Theta &= \lub \{\, X \subseteq S \mid X \subseteq \Theta(X) \,\} \ . \label{eq:tarski:gfp}
\end{align}
It is shown in \Propn~\ref{prop:EC:W} below that the greatest fixpoint corresponds to the stutter step formula $P \weakuntil T$, and in \Propn~\ref{prop:EC:U} that the least fixpoint corresponds to $P \until T$.

A slight drawback of~\eqref{eq:tarski:lfp} and~\eqref{eq:tarski:gfp} is that they are not constructive. One construction that finds the fixpoints is transfinite recursion. For an operator $\Theta \colon 2^S \to 2^S$, its transfinite iterates $\ThetaUp\alpha$ and $\ThetaDown\alpha$ are defined for all ordinals~$\alpha$:
\begin{align*}
  \ThetaUp 0 &= \emptyset \\
    \ThetaUp{(\beta+1)} &= \Theta(\ThetaUp\beta) && \textrm{for successor ordinals;} \\
    \ThetaUp{\bigg(\lub_{\beta < \alpha} \beta\bigg)} &= \lub_{\beta < \alpha} (\ThetaUp \beta) && \textrm{for limit ordinals.} \\[\the\smallskipamount]
  \ThetaDown 0 &= S \ ; \\
    \ThetaDown{(\beta+1)} &= \Theta(\ThetaDown\beta) && \textrm{for successor ordinals;} \\
    \ThetaDown{\bigg(\lub_{\beta < \alpha} \beta\bigg)} &= \glb_{\beta < \alpha} (\ThetaDown \beta) && \textrm{for limit ordinals.}
\end{align*}
When $\Theta$ is \monotonic, as it is here, the limits of $\ThetaUp\alpha$ and $\ThetaDown\alpha$ give the least and greatest fixpoint, respectively~\citep{Cou:79}.
In summary,
\begin{align}
    \lfp \Theta &= \glb \{\, X \subseteq S \mid \Theta(X) \subseteq X \,\} =  \ThetaUp{\big(\lub_\alpha \alpha\big)}
    \label{eq:lfp:glb} \\
    \gfp \Theta &= \lub \{\, X \subseteq S \mid X \subseteq \Theta(X) \,\} = \ThetaDown{\big(\lub_\alpha  \alpha\big)}
    \label{eq:gfp:lub}
\end{align}
The following \propn~\ref{prop:EC:W} shows that the sets $\EC[G]{P
  \weakuntil T}$ and $\ECS[G]{P \weakuntil T}$ from which a weak until
formula can be enforced by a general or positional controller are both
equal to the greatest fixpoint of~$\Theta$. In the finite-state case, this
makes it possible to compute these sets using~\eqref{eq:gfp:lub}, while
termination is not guaranteed for infinite state spaces.

\begin{proposition}
  \label{prop:EC:W}
  Let $G = \tsystem$ be a \dlfree
  transition system, and let $\psi \equiv P \weakuntil T$ be a stutter step formula for~$G$. Then
  \begin{equation}
    \EC[G]{\psi} =
    \ECS[G]{\psi} =
    \gfp \Theta \ .
  \end{equation}%
  Additionally, there is a \memoryless\ controller $\overlinit{C}\colon S \to 2^\Sigma$
  such that $\langle \overlinit{C}/G,s\rangle \models \psi$ for all $s \in
  \ECS[G]{\psi}$.
\end{proposition}

\begin{proof}
By~\eqref{eq:gfp:lub}, the greatest fixpoint of~$\Theta$ can written as
$\check{W} = \gfp \Theta = \lub \{ X \subseteq S \mid X \subseteq \Theta(X) \}$.
Then it is to be shown that $W = \EC[G]{\psi}$ and $W' =
\ECS[G]{\psi}$ are both equal to~$\check{W}$. Noting that
$\ECS[G]\psi \subseteq \EC[G]\psi$ for any formula~$\psi$, it is clear that
$W' \subseteq W$. Therefore, it is enough to show $W \subseteq
\check{W} \subseteq W'$.

\begin{itemize}
\item
  To show $W \subseteq \check{W} = \lub \{\, X \subseteq S \mid X \subseteq
  \Theta(X) \,\}$, it is enough to show $W \subseteq \Theta(W)$. Therefore
  let $w \in W = \EC[G]\psi$. Then there exists a controller
  $C$ such that $\langle C/G, w\rangle \models \psi$. If $w \in
  T$, it already holds that $w \in T \cup (P \cap \Pre(W)) = \Theta(W)$. If
  $w \notin T$, then $\langle C/G, w\rangle \models \psi \equiv P
  \weakuntil T$ implies
  that $w \in P$ and there exists $\sigma \in C(w)$ such that $\emptyset
  \neq \Post(w, \sigma) \subseteq \EC[G]{\psi} = W$, i.e., $w \in
  P \cap \Pre(W) \subseteq T \cup (P \cap \Pre(W)) = \Theta(W)$.
\item
  To show $\check{W} \subseteq W'$, first note that 
  \begin{equation}
    \label{eq:EC:AB}
    \check{W} = \Theta(\check{W}) = T \cup (P \cap \Pre(\check{W}))
            \subseteq P \cup T \ .
  \end{equation}
  Also note that $P \cap T = \emptyset$ since $\psi$ is a stutter step formula.
  Define a \memoryless\ controller $\overlinit{C} \colon S \to 2^\Sigma$ such that
  \begin{equation}
    \label{eq:EC:C}
    \overlinit{C}(s) = \begin{cases}
             \{\, \sigma \in \Sigma \mid \Post(s,\sigma) \subseteq \check{W} \,\},
             & \text{if}\ s \in \check{W} \setminus T; \\
             \Sigma, & \text{otherwise}.
    \end{cases}
  \end{equation}
  For $s \in \check{W} \setminus T$, it holds that $s \in \check{W} =
  \Theta(\check{W}) = T \cup (P \cap \Pre(\check{W}))$ and then $s \in P \cap
  \Pre(\check{W}) \subseteq \Pre(\check{W})$, so there exists $\sigma \in
  \Sigma$ such that $\emptyset \neq \Post(s,\sigma) \subseteq \check{W}$.
  Therefore, $\overlinit{C}$ is a \dlfree controller.
  
  It remains to be shown that $\langle \overlinit{C}/G,s \rangle \models \psi$ for all $s \in \check{W}$. Let $\pi = s_0s_1\cdots \in \Frags^{\omega}(\overlinit{C}/G)$ be an
  infinite path fragment in~$\overlinit{C}/G$ such that $s_0 \in \check{W}$. It is shown by
  induction on~$n$ that $s_0s_1\cdots s_n \models P \until T$ or
  $s_0s_1\cdots s_n \in \check{W}^*$.

  First $s_0 \in \check{W}^*$ since $s_0 \in \check{W}$ by assumption.

  Next consider $s_0s_1\cdots s_n s_{n+1}$, where by inductive assumption
  $s_0s_1\cdots s_n \models P \until T$ or $s_0s_1\cdots s_n \in \check{W}^*$. If $s_0s_1\cdots s_n \models P \until
  T$, it immediately holds that $s_0s_1\cdots s_n s_{n+1} \models P \until T$. Otherwise $s_0s_1\cdots s_n \not\models P \until T$ and
  $s_0s_1\cdots s_n \in (\check{W} \setminus T)^*$, so it follows
  from~\eqref{eq:EC:AB} that $s_0s_1\cdots s_n \in P^*$. Also $s_n
  \in \check{W} \setminus T$, and since $s_0s_1\cdots s_n s_{n+1}$ is a path fragment in~$\overlinit{C}/G$
  there exists $\sigma \in \overlinit{C}(s_n)$ such that $s_{n+1} \in \Post(s_n,\sigma)
  \subseteq \check{W}$ by~\eqref{eq:EC:C}. This is enough to show
  $s_0s_1\cdots s_n s_{n+1} \in \check{W}^*$.

  This completes the induction. Now consider an infinite path fragment~$\pi$
  in~$\overlinit{C}/G$ starting from $s \in \check{W}$. If some prefix of~$\pi$ satisfies
  $P \until T$, then $\pi \models P \weakuntil T \equiv \psi$. Otherwise $\pi \in
  \check{W}^\omega$ and no prefix of~$\pi$ satisfies $P \until T$, so
  $\pi \in (\check{W} \setminus T)^\omega$. Then it follows
  from~\eqref{eq:EC:AB} that $\pi \in P^\omega$, which implies $\pi \models
  P \weakuntil T \equiv \psi$. As $\pi$ was chosen to be an arbitrary path fragment starting at
  an arbitrary $s \in \check{W}$, it follows that $\check{W} \subseteq \ECS[G]{\psi} = W'$.
\item
  Additionally, it follows from the construction of~$\overlinit{C}$ in the above proof
  of $\check{W} \subseteq W'$ that $\langle \overlinit{C}/G,s\rangle \models \psi$ for all $s \in \check{W} = \ECS[G]{\psi}$.
  \QED 
\end{itemize}
\end{proof}

\Propn~\ref{prop:EC:U} is similar to \propn~\ref{prop:EC:W} and relates
$\EC[G]{P \until T}$ and $\ECS[G]{P \until T}$ to the least fixpoint
of~$\Theta$.

\begin{proposition}
  \label{prop:EC:U}
  Let $G = \tsystem$ be a \dlfree transition system, and let $\psi \equiv P \until T$ be a stutter step formula. Then
  \begin{equation}
    \EC[G]{\psi} =
    \ECS[G]{\psi} =
    \lfp \Theta \ .
  \end{equation}
  Additionally, there is a \memoryless\ controller $\overlinit{C}\colon S \to 2^\Sigma$
  such that $\langle \overlinit{C}/G,s\rangle \models \psi$ for all $s \in
  \ECS[G]{\psi}$.
\end{proposition}

\begin{proof}
Denote the least fixpoint of~$\Theta$ by $\hat{U} = \lfp \Theta$.
It is to be shown that $U =
\EC[G]{\psi}$ and $U' = \ECS[G]{\psi}$ are both equal
to~$\hat{U}$. Noting that $\ECS[G]\psi \subseteq \EC[G]\psi$ for any
formula~$\psi$, it is clear that $U' \subseteq U$. Then it is enough
to show $U \subseteq \hat{U} \subseteq U'$.

\begin{itemize}
\item
  To show $U \subseteq \hat{U}$, assume the contrary, i.e., there exists
  $u_0 \in U$ such that $u_0 \notin \hat{U}$. As $u_0 \in U = \EC[G]{\psi}$, there exists a controller~$C$ such that $\langle
  C/G,u_0\rangle \models \psi$.
  Note that $u_0 \notin \hat{U} = \Theta(\hat{U}) = T \cup (P \cap
  \Pre(\hat{U})) \supseteq T$, i.e., $u_0 \notin T$. Then $\langle
  C/G,u_0\rangle \models \psi \equiv P \until T$ implies $u_0 \in P$, i.e.,
  $u_0 \in (U \cap P) \setminus \hat{U}$.
  From $\langle C/G,u_0\rangle \models P \until T$ and $u_0 \notin T$, it
  follows that there exists $\sigma_0 \in C(u_0)$ such that
  $\Post(u_0,\sigma_0) \neq \emptyset$ and $\langle C/G, u_0u_1\rangle
  \models P \until T$ for all $u_1 \in \Post(u_0,\sigma_0)$.
  As $u_0 \notin \hat{U} = \Theta(\hat{U}) = T \cup (P \cap \Pre(\hat{U}))
  \supseteq P \cap \Pre(\hat{U})$ and $u_0 \in P$, it follows that $u_0
  \notin \Pre(\hat{U})$. This means that there does not exist $\sigma \in
  \Sigma$ such that $\emptyset \neq \Post(u_0,\sigma) \subseteq \hat{U}$.
  Then for above $\sigma_0 \in C(u_0)$ there exists $u_1 \in
  \Post(u_0,\sigma_0)$ such that $u_1 \notin \hat{U}$. Then $u_1 \notin
  \hat{U} \supseteq T$ and $\langle C/G, u_0u_1\rangle \models P \until T$
  together imply $u_1 \in \EC[G]{P \until T} \cap P = U \cap P$. Thus $u_1
  \in (U \cap P) \setminus \hat{U}$.
  By induction, it is possible to construct an infinite sequence of states
  $u_i \in (U \cap P) \setminus \hat{U} \subseteq P \setminus T$ and labels
  $\sigma_i \in C(u_0 \cdots u_i)$ such that $u_{i+1} \in
  \Post(u_i,\sigma_i)$ for all $i \geq 0$. But then it does not hold that
  $\langle C/G,u_0\rangle \models P \until T \equiv \psi$, which
  contradicts the choice of~$C$ as a controller enforcing this property.

\item
  To show $\hat{U} \subseteq U'$, recall that $\hat{U} = \lfp\Theta =
  \lub_\alpha (\ThetaUp\alpha)$ by~\eqref{eq:lfp:glb}. For each
  ordinal~$\alpha$, define a set
  \begin{equation}
    T_\alpha = (\ThetaUp\alpha) \setminus
               \lub_{\beta<\alpha} (\ThetaUp\beta).
  \end{equation}
  For successor ordinals~$\alpha$, the set~$T_\alpha$ contains precisely the elements added in the corresponding iteration of~$\Theta$, while $T_\alpha = \emptyset$ for limit ordinals. The nonempty sets $T_\alpha$ form a partition of~$\hat{U}$ by~\eqref{eq:lfp:glb}, i.e., for each $s \in \hat{U}$ there exists exactly one ordinal~$\alpha$ such that $s \in T_\alpha$. Therefore, the following is a well-defined  \memoryless\ controller $\overlinit{C}\colon S \to 2^\Sigma$:
  \begin{equation}
    \label{eq:tildeU:C}
    \overlinit{C}(s) = \begin{cases}
             \mathrlap{\{\, \sigma \in \Sigma \mid \Post(s,\sigma) \subseteq
                  \lub_{\beta<\alpha} (\ThetaUp\beta) \,\},} & \\
             & \kern3em\text{if}\ s \in T_\alpha \setminus T; \\
             \Sigma, & \kern3em\text{if}\ s \notin \hat{U} \setminus T.
    \end{cases}
  \end{equation}
  It is shown by transfinite induction that, for all $s \in T_\alpha
  \setminus T$,
  \begin{align}
    & \exists \sigma \in \overlinit{C}(s)
      \ldotp \Post(s,\sigma) \neq \emptyset \ ; \label{eq:tildeU:enables} \\
    & \langle \overlinit{C}/G, s\rangle \models \psi \ .
      \label{eq:tildeU:models}
  \end{align}
  Condition~\eqref{eq:tildeU:enables} implies that $\overlinit{C}$ is a
  \dlfree controller, and as $\langle \overlinit{C}/G, s\rangle
  \models P \until T \equiv \psi$ holds trivially for $s \in T$,
  condition~\eqref{eq:tildeU:models} implies $\langle \overlinit{C}/G,
  s\rangle \models \psi$ for all $s \in \hat{U}$. This will be enough to
  show the claim $\hat{U} \subseteq U'$.
  So let $\alpha$ be an ordinal, and assume as inductive assumption that \eqref{eq:tildeU:enables} and~\eqref{eq:tildeU:models} hold for all ordinals $\beta < \alpha$ and all states $s' \in T_\beta \setminus T$. Let $s \in T_\alpha \setminus T$ and consider two cases.
    
  If $\alpha = \gamma + 1$ is a successor ordinal, then $s \in T_\alpha
  \subseteq \ThetaUp\alpha = \Theta(\ThetaUp\gamma) = T \cup (P \cap
  \Pre(\ThetaUp\gamma))$. As $s \notin T$, it follows that $s \in P$ and
  there exists $\sigma \in \Sigma$ such that $\emptyset \neq
  \Post(s,\sigma) \subseteq \ThetaUp\gamma \subseteq \lub_{\beta<\alpha}
  (\ThetaUp\beta)$. Then $\sigma \in \overlinit{C}(s)$ by~\eqref{eq:tildeU:C},
  showing~\eqref{eq:tildeU:enables}. Also by~\eqref{eq:tildeU:C}, for all
  $\sigma \in \overlinit{C}(s)$ it holds that $\Post(s,\sigma) \subseteq
  \lub_{\beta<\alpha} (\ThetaUp\beta) = \lub_{\beta<\alpha} T_\beta$. Then
  by inductive assumption $\langle \overlinit{C}/G, s'\rangle \models \psi
  \equiv P \until T$ for all states $s' \in \Post(s,\sigma)$ with $\sigma
  \in \overlinit{C}(s)$. Together with $s \in P$, this implies $\langle
  \overlinit{C}/G, s\rangle \models P \until T \equiv \psi$ and
  shows~\eqref{eq:tildeU:models}.
  
  If $\alpha = \lub_{\beta<\alpha} \beta$ is a limit ordinal, then
  $T_\alpha = (\ThetaUp\alpha) \setminus \lub_{\beta<\alpha}(\ThetaUp\beta) =
  \lub_{\beta<\alpha}(\ThetaUp\beta) \setminus
  \lub_{\beta<\alpha}(\ThetaUp\beta) = \emptyset$ and there is nothing to
  be shown.
\item
  Additionally, it follows from the construction of~$\overlinit{C}$ in the above proof
  of $\hat{U} \subseteq U'$ that $\langle \overlinit{C}/G,s\rangle \models \psi$
  for all $s \in \hat{U} = \ECS[G]{\psi}$. \QED
\end{itemize}
\end{proof}

\section{Robust stutter bisimulation quotient}
\label{app:quotient}

This section gives the necessary details to prove
\Thm~\ref{thm:quotient:robust:stutter:bisimilar}, which shows that a
transition system is robust stutter bisimilar to its quotient modulo a
\rdsb. The equivalence of a transition system~$G$ with its quotient~$G/\RR$
will be established by considering the union transition system $G \cup
(G/\RR)$, and constructing an equivalence relation on its state set $S \cup
(S/\RR)$. For this construction, an equivalence relation $\RR \subseteq S
\times S$ is extended to a relation on $S \cup (S/\RR)$ in a natural way by
defining that every element $s \in S$ is also equivalent to its equivalence
class $[s]_\RR \in S/\RR$.

\begin{definition}
\label{def:hatR}
Let $\RR \subseteq S \times S$ be an equivalence relation.
The \emph{extension} of \RR\ to $S \cup (S/\RR)$
is the smallest equivalence relation that includes
\begin{equation}
	\RR \cup \{\, (s, [s]_\RR) \mid s \in S \,\} \ .
\end{equation}
\end{definition}

\begin{remark}
\label{rem:hatR}
An explicit construction of the extension of \RR\ to $S \cup (S/\RR)$ is
given by
\begin{align*}
  \hat\RR &=
  \RR \cup
  \{\, (\tilde{s}, \tilde{s}) \mid \tilde{s} \in S/\RR \,\} \cup
  {} \\ & \hphantom{{}={}}
  \{\, (s, [s]_\RR) \mid s \in S \,\} \cup
  \{\, ([s]_\RR, s) \mid s \in S \,\} \ .
\end{align*}
That is, the extension includes the relation~\RR\ and the identity relation
on~$S/\RR$, plus the pairs $(s, [s]_\RR)$ and $([s]_\RR, s)$ that relate an
element $s \in S$ to the equivalence class~$[s]_\RR \in S/\RR$ containing it.
\end{remark}

Conversely, an equivalence relation $S \cup (S/\RR)$ can be
\emph{restricted} to relations on the subsets $S$ and~$S/\RR$, and such
restriction results in equivalence relations.

\begin{definition}
	\label{def:restrict}
	Let $\RR \subseteq S \times S$ be an equivalence relation, and let $S_1
	\subseteq S$. The \emph{restriction} of $\RR$ to~$S_1$,
	denoted~$\RR|S_1$, is $\RR|S_1 = \RR \cap (S_1 \times S_1)$.
\end{definition}

\begin{lemma}\label{lem:restrict}
Let $S_1$ and $S_2$ be disjoint sets and let $S = S_1 \cup S_2$, and let
$\RR$ be an equivalence relation on~$S$. Then
\begin{enumerate}
	\item
	\label{it:restrict:equivalence}
	$\RR|S_1$ is an equivalence relation on~$S_1$.
	\item 
	\label{it:restrict:relation}
	For all $s_1 \in S_1$ it holds that $[s_1]_{\RR|S_1} = [s_1]_\RR \cap S_1$.
	\item
	\label{it:restrict:superblock}
	$T \in \SB(\RR)$ implies $T \cap S_1 \in \SB(\RR|S_1)$.
\end{enumerate}
\end{lemma}

\begin{proof}
\ref{it:restrict:equivalence}
As \RR\ is reflexive, $(s, s) \in \RR \cap (S_1 \times S_1) = \RR|S_1$
for any $s \in S_1$, so $\RR|S_1$ is reflexive.
If $(s, t) \in \RR|S_1$ then $(s, t) \in \RR$ and $s, t \in S_1$. As $\RR$
is symmetric, it follows that $(t, s) \in \RR \cap (S_1 \times S_1) =
\RR|S_1$, so $\RR|S_1$ is symmetric.
If $(s, t) \in \RR|S_1$ and $(t, u) \in \RR|S_1$, then $(s, t), (t, u) \in
\RR$ and $s, t, u \in S_1$. As $\RR$ is transitive, it likewise follows
that $(s, u) \in \RR|S_1$, so $\RR|S_1$ is transitive.

\ref{it:restrict:relation}
Let $s_1 \in S_1$ and consider an arbitrary element $s \in S$.
Then $s \in [s_1]_\RR \cap S_1$
is equivalent to $(s_1,s) \in \RR \cap (S_1 \times S_1) = \RR|S_1$,
which in turn is equivalent to $s \in [s_1]_{\RR|S_1}$.

\ref{it:restrict:superblock}
Let $T \in \SB(\RR)$. 
To show $T \cap S_1 \in \SB(\RR|S_1)$, following \defn~\ref{def:superblock},
it is to be shown that, if $s \in T \cap S_1$ and $(s, t) \in \RR|S_1$,
then also $t \in T \cap S_1$.
So let $s \in T \cap S_1$ and $(s, t) \in \RR|S_1$.
From $(s, t) \in \RR|S_1 \subseteq \RR$ and $s \in T \in \SB(\RR)$, it
follows that $t \in T$.
Further, $(s, t) \in \RR|S_1$ implies $t \in S_1$,
and thus $t \in T \cap S_1$.
\end{proof}

The following lemma shows how a stutter step formula for union transition
system is related to a stutter step formula on only one of the systems in
the union.

\begin{lemma}
\label{lem:Gunion}
Let $G_i = \tsystem[i]$ for $i=1,2$ be two \dlfree transition systems
with disjoint state sets, and let $G = G_1 \cup G_2 = \tsystem$.
For any stutter step formula $P \anyuntil T$ with $P,T \subseteq S$ it
holds that
\begin{equation}
	\ECS[G]{P \anyuntil T} \cap S_1 =
	\ECS[G_1]{(P \cap S_1) \anyuntil (T \cap S_1)} \ .
\end{equation}
\end{lemma}

\begin{proof}
First, let $s_1 \in \ECS[G]{P \anyuntil T} \cap S_1$.
Then $s_1 \in S_1$, and there exists a \dlfree
\memoryless\ controller $\overlinit{C} 
\colon S \to 2^\Sigma$ such that $\langle \overlinit{C}/G, s_1 \rangle
\models P \anyuntil T$.
Define a \memoryless\ controller $\overlinit{C}_1 \colon S_1 \to 2^{\Sigma_1}$
for~$G_1$ with $\overlinit{C}_1(s) = \overlinit{C}(s) \cap \Sigma_1$ for $s \in S_1$.
Note that $\overlinit{C}(s) \cap \Sigma_1 \neq \emptyset$ for $s \in S_1$ as
$\overlinit{C}$ is \dlfree and by the construction of~$G$ as disjoint
union $G_1 \cup G_2$, only labels in~$\Sigma_1$ are enabled at $s
\in S_1$. Therefore $\overlinit{C}_1$ is \dlfree.
To show $\langle \overlinit{C}_1/G_1, s_1 \rangle \models (P \cap S_1) \anyuntil
(T \cap S_1)$, assume $\pi \in \Frags^\omega(\overlinit{C}_1, G_1)$ such that $s_1
\sqsubseteq \pi$.
Then $\pi \in \Frags^\omega(G_1) \subseteq \Frags^\omega(G)$.
Also $\pi \in S_1^\omega$, which means that only labels in~$\Sigma_1$ are
enabled at states of~$\pi$, so that $\pi \in \Frags^\omega(\overlinit{C},G)$ by
construction of~$\overlinit{C}_1$.
As $\langle \overlinit{C}_1/G_1, s_1 \rangle \models P \anyuntil T$, it follows
that $\pi \models P \anyuntil T$. Given $\pi \in S_1^\omega$, this
implies $\pi \models (P \cap S_1) \anyuntil (T \cap S_1)$.
Hence, $s_1 \in \ECS[G_1]{(P \cap S_1) \anyuntil (T \cap S_1)}$ and therefore $\ECS[G]{P \anyuntil T} \cap S_1 \subseteq \ECS[G_1]{(P \cap S_1) \anyuntil (T \cap S_1)}$.

Conversely, let $s_1 \in \ECS[G_1]{(P \cap S_1) \anyuntil (T \cap S_1)}$.
Then there exists a \dlfree \memoryless\ controller~$\overlinit{C}_1$ such
that $\langle \overlinit{C}_1/G_1, s_1 \rangle \models (P \cap S_1)
\anyuntil (T \cap S_1)$.
Clearly either $s_1 \in (P \cap S_1) \subseteq S_1$ or $s_1 \in (T \cap
S_1) \subseteq S_1$, and thus $s_1 \in S_1$.
Define a \memoryless\ controller $\overlinit{C} \colon S \to 2^\Sigma$ for~$G$
such that $\overlinit{C}(s) = \overlinit{C}_1(s)$ if $s \in S_1$, and $\overlinit{C}(s) =
\Sigma$ otherwise.
$\overlinit{C}$ behaves like $\overlinit{C}_1$ for all states in~$S_1$.
As $\overlinit{C}_1$ and~$G_2$ are \dlfree, $\overlinit{C}$ is also
\dlfree.
To show $\langle \overlinit{C}/G, s_1 \rangle \models P \anyuntil T$, let $\pi
\in \Frags^\omega(\overlinit{C}, G)$ such that $s_1 \sqsubseteq \pi$.
Then $\pi \in \Frags^\omega(G)$, and as
$G = G_1 \cup G_2$ is composed of two disjoint transition systems, all
the states on the path~$\pi$ must either be in~$S_1$ or in~$S_2$.
Then, since $s_1 \sqsubseteq \pi$ with $s_1 \in S_1$, it follows that $\pi
\in S_1^\omega$ and $\pi \in \Frags^\omega(G_1)$.
By construction also $\overlinit{C}(s) = \overlinit{C}_1(s)$ for all states~$s$
on the path fragment~$\pi$, and thus $\pi \in \Frags^\omega(\overlinit{C}_1,G_1)$.
As $\langle \overlinit{C}_1/G_1, s_1 \rangle \models (P \cap S_1) \anyuntil (T
\cap S_1)$, it follows that $\pi \models (P \cap S_1) \anyuntil (T \cap
S_1)$, and this implies $\pi \models P \anyuntil T$.
Hence, $s_1 \in \ECS[G]{P \anyuntil T}$ and therefore $\ECS[G_1]{(P \cap S_1) \anyuntil (T \cap S_1)} \subseteq \ECS[G]{P \anyuntil T}$.
\end{proof}

Using the above lemmas, it is now possible to prove
\thm~\ref{thm:quotient:robust:stutter:bisimilar}, by showing that the
extension of the relation~\RR\ to states of the combined system $G \cup
(G/\RR)$ is a robust stutter bisimulation.

\begin{repeattheorem}{thm:quotient:robust:stutter:bisimilar}
Let $G$ be a transition system and \RR\ be a robust stutter simulation on~$G$. Then $G \approx G/\RR$.
\end{repeattheorem}

\begin{proof}
Let $G = \tsystem$, and let $\hat{\RR}$ be the extension of \RR\ to $S \cup (S/\RR)$.
It will be shown that $\hat{\RR}$ is a robust stutter bisimulation for the transition system $\hat{G} = G \cup (G/\RR)$.
As~$\hat\RR$ is an equivalence relation by \defn~\ref{def:hatR}, it remains
to prove conditions \ref{it:rsb:label}
and~\ref{it:rsb:equivalent}
of \defn~\ref{def:robust:stutter:bisimulation}.
Consider a pair $(s,t) \in \hat\RR$.
There are four cases to consider depending on whether $s$ or~$t$ are elements of~$S$ or
of~$S/\RR$.
\begin{itemize}
\item $(s, t) \in \hat{\RR}$ where $s, t \in S$. Then $(s, t) \in \RR$ by
  \defn~\ref{def:hatR}.
  As \RR\ is a robust stutter bisimulation on~$G$, it holds by
  \defn~\ref{def:robust:stutter:bisimulation} that \ref{it:rsb:label}
  $L(s) = L(t)$ and \ref{it:rsb:equivalent} $s \in \ECS[G]\psi$ implies $t
  \in \ECS[G]\psi$ for every \RR-superblock step formula $\psi$
  from~$[s]_\RR$.
  Condition~\ref{it:rsb:label} for $\hat\RR$ and~$\hat{G}$ follows immediately.
  For condition~\ref{it:rsb:equivalent}, let $\hat\psi \equiv
  [s]_{\hat{\RR}} \anyuntil \hat{T}$ for some $\hat{T} \in
  \SB(\hat{\RR})$ such that $[s]_{\hat\RR} \cap \hat{T} = \emptyset$,
  and let $s \in \ECS[\hat{G}]{\hat\psi}$. It is to be shown that $t \in
  \ECS[\hat{G}]{\hat\psi}$.

  First note that $[s]_{\RR\vphantom{\hat\RR}} = [s]_{\hat\RR} \cap S$ by
  \lemm~\ref{lem:restrict}~\ref{it:restrict:relation}, and $\hat{T} \cap S
  \in \SB(\RR)$ by \lemm~\ref{lem:restrict}~\ref{it:restrict:superblock}.
  Let $T = \hat{T} \cap S$.
  Then $[s]_\RR \cap T = ([s]_{\hat\RR} \cap S) \cap
  (\hat{T} \cap S) \subseteq [s]_{\hat\RR} \cap \hat{T} = \emptyset$, so
  that $\psi \equiv [s]_\RR \anyuntil T$ is an
  \RR-superblock step formula from~$[s]_\RR$.
  Also, $s \in \ECS[\hat{G}]{\hat\psi}$ means that there exists a \dlfree
  controller for~$\hat{G}$, so $\hat{G}$, $G$, and~$G/\RR$ are \dlfree.
  By \lemm~\ref{lem:Gunion}, it follows that $s \in \ECS[\hat{G}]{\hat\psi}
  \cap S = \ECS[\hat{G}]{[s]_{\hat\RR} \anyuntil \hat{T}} \cap S =
  \ECS[G]{([s]_{\hat\RR} \cap S) \anyuntil (\hat{T} \cap S)} = \ECS[G]{[s]_\RR
  \anyuntil T} = \ECS[G]\psi$.
  Since $\psi$ is an \RR-superblock step formula from~$[s]_\RR$ and
  \RR\ is a \rdsb, this
  implies $t \in \ECS[G]\psi = \ECS[G]{[s]_\RR \anyuntil T}
  = \ECS[G]{([s]_{\hat\RR} \cap S) \anyuntil (\hat{T} \cap S)} =
  \ECS[\hat{G}]{[s]_{\hat\RR} \anyuntil \hat{T}} \cap S =
  \ECS[\hat{G}]{\hat\psi} \cap S \subseteq \ECS[\hat{G}]{\hat\psi}$, again
  by \lemm~\ref{lem:Gunion}.
	
\item $(s, \tilde{s}) \in \hat{\RR}$ where $s \in S$ and $\tilde{s} \in S/\RR$.
  By \defn~\ref{def:hatR} it is clear that $s \in \tilde{s}$, which means
  that $\tilde{s} = [s]_\RR$.
  Note that $L_\RR(\tilde{s}) = L_{\RR}([s]_\RR) = L(s)$ by
  \defn~\ref{def:quotient:system}, showing condition~\ref{it:rsb:label} of
  \defn~\ref{def:robust:stutter:bisimulation}.
  For condition~\ref{it:rsb:equivalent}, let $\hat\psi \equiv
  [s]_{\hat{\RR}} \anyuntil \hat{T}$ for some $\hat{T} \in
  \SB(\hat{\RR})$ such that $[s]_{\hat\RR} \cap \hat{T} =
  \emptyset$, and let $s \in \ECS[\hat{G}]{\hat\psi}$. It is to be shown
  that $\tilde{s} \in \ECS[\hat{G}]{\hat\psi}$.

  Firstly, $s \in \ECS[\hat{G}]{\hat\psi}$ implies that there exists a
  \dlfree controller for~$\hat{G}$, and thus $\hat{G}$, $G$, and~$G/\RR$
  are \dlfree.
  Furthermore,
  \begin{align*}
    s & \in \ECS[\hat{G}]{\hat\psi} \cap S \\
      & = \ECS[\hat{G}]{[s]_{\hat{\RR}} \anyuntil \hat{T}} \cap S \\
      & = \ECS[G]{([s]_{\hat{\RR}} \cap S) \anyuntil (\hat{T} \cap S)}
      &   \text{(by \lemm~\ref{lem:Gunion})} \\
      & = \ECS[G]{[s]_\RR \anyuntil (\hat{T} \cap S)}
      &   \kern-1em
          \text{(by \lemm~\ref{lem:restrict}~\ref{it:restrict:relation})} \\
      & = \ECS[G]{\tilde{s} \anyuntil (\hat{T} \cap S)} \ .
      &   \text{(as $\tilde{s} = [s]_\RR$)}
  \end{align*}
  Let $T = \hat{T} \cap S$ and $\psi \equiv \tilde{s} \anyuntil T$ so that
  $s \in \ECS[G]{\psi}$.
  Note that $T = \hat{T} \cap S \in \SB(\RR)$ by
  \lemm~\ref{lem:restrict}~\ref{it:restrict:superblock} and $\tilde{s} =
  [s]_\RR = [s]_{\hat\RR} \cap S$ by
  \lemm~\ref{lem:restrict}~\ref{it:restrict:relation}, and then $\tilde{s}
  \cap T = ([s]_{\hat\RR} \cap S) \cap (\hat{T} \cap S)
  \subseteq [s]_{\hat\RR} \cap \hat{T} = \emptyset$. Thus $\psi$ is an
  \RR-superblock step formula from~$\tilde{s}$.
  As \RR\ is a robust stutter bisimulation on~$G$ and $s \in
  \ECS[G]{\psi}$, it follows that $\tilde{s} = [s]_\RR \subseteq
  \ECS[G]{\psi}$. Then $\psi \in \Sigma_\RR$ by~\eqref{eq:quotient:Sigma}.
  Define a \memoryless\ controller $\overlinit{C} \colon S/\RR \to
  2^{\Sigma_\RR}$ with $\overlinit{C}(\tilde{s}) = \{ \psi \}$ and
  $\overlinit{C}(t) = \Sigma_\RR$ for $t \neq \tilde{s}$, i.e., the only
  transitions enabled by~$\overlinit{C}$ from the state~$\tilde{s}$ are
  those labelled~$\psi$.

  Let $T_\RR = \hat{T} \cap (S/\RR)$, and consider two cases depending on
  whether $\anyuntil$ in $\psi$ and~$\hat\psi$ is $\until$ or~$\weakuntil$:
  \begin{itemize}
  \item
    If $\mathord{\anyuntil} = \mathord{\until}$ and thus $\psi \equiv
    \tilde{s} \until T$, it follows from \eqref{eq:quotient:delta} that
    $(\tilde{s}, \psi, T') \in \delta_\RR$ implies $T' \in S/\RR$ and $T'
    \subseteq T$.
    From $T' \in S/\RR$ it follows that $T' = [t']_\RR$ for some $t' \in
    S$, and then $t' \in [t']_\RR = T' \subseteq T = \hat{T} \cap S
    \subseteq \hat{T}$, and then also $T' \in \hat{T}$ because $\hat{T}
    \in \SB(\hat\RR)$ and $(t',T') = (t',[t']_\RR) \in \hat\RR$ by
    \defn~\ref{def:hatR}.
    This means $T' \in \hat{T} \cap (S/\RR) = T_\RR$, and it follows by
    construction of $\overlinit{C}(\tilde{s}) = \psi$ that $\langle
    \overlinit{C}/(G/\RR), \tilde{s}\rangle \models \{\tilde{s}\} \until
    T_\RR$.
  \item
    If $\mathord{\anyuntil} = \mathord{\weakuntil}$ and thus $\psi \equiv
    \tilde{s} \weakuntil T$, it likewise follows
    from~\eqref{eq:quotient:delta} that $(\tilde{s},
    \psi, T') \in \delta_\RR$ implies $T' \in T_\RR \cup
    \{\tilde{s}\}$, and therefore $\langle \overlinit{C}/(G/\RR),
    \tilde{s}\rangle \models \{\tilde{s}\} \weakuntil T_\RR$.
  \end{itemize}
  In both cases it is clear that $\overlinit{C}/G$ is \dlfree and $\langle
  \overlinit{C}/(G/\RR), \tilde{s}\rangle \models \{\tilde{s}\} \anyuntil
  T_\RR$. It follows that
  \begin{align*}
    \tilde{s} & \in \ECS[G/\RR]{\{\tilde{s}\} \anyuntil T_\RR} \\
              & = \ECS[G/\RR]{(\{\tilde{s}\} \cap (S/\RR))
                               \anyuntil (\hat{T} \cap (S/\RR))} \span \\
              & = \ECS[\hat{G}]{\{\tilde{s}\} \anyuntil \hat{T}} \cap (S/\RR)
              & \text{(by \lemm~\ref{lem:Gunion})} \\
              & \subseteq \ECS[\hat{G}]{[\tilde{s}]_{\hat\RR}
                                        \anyuntil \hat{T}} \\
              & = \ECS[\hat{G}]{[s]_{\hat\RR} \anyuntil \hat{T}}
              & \text{(as $(s,\tilde{s}) \in \hat\RR$)} \\
              & = \ECS[\hat{G}]{\hat\psi} \ .
  \end{align*}      
	
\item $(\tilde{s},s) \in \hat{\RR}$ where $\tilde{s} \in S/\RR$ and $s \in
  S$. By \defn~\ref{def:hatR} it is clear that $s \in \tilde{s}$, which
  means that $\tilde{s} = [s]_\RR$.
  Note that $L_\RR(s) = L_{\RR}([s]_\RR) = L(\tilde{s})$ by
  \defn~\ref{def:quotient:system}, showing condition~\ref{it:rsb:label} of
  \defn~\ref{def:robust:stutter:bisimulation}.
  For condition~\ref{it:rsb:equivalent}, let $\hat\psi \equiv
  [\tilde{s}]_{\hat{\RR}} \anyuntil \hat{T}$ for some $\hat{T} \in
  \SB(\hat{\RR})$ such that $[\tilde{s}]_{\hat{\RR}} \cap \hat{T} =
  \emptyset$, and let $\tilde{s} \in \ECS[\hat{G}]{\hat\psi}$. It is to be
  shown that $s \in \ECS[\hat{G}]{\hat\psi}$.

  Again, $\tilde{s} \in \ECS[\hat{G}]{\hat\psi}$ implies that there is a
  \dlfree controller for~$\hat{G}$, and thus $\hat{G}$, $G$, and~$G/\RR$
  are \dlfree. Consider $T_\RR = \hat{T} \cap (S/\RR)$. By
  \lemm~\ref{lem:Gunion},
  \lemm~\ref{lem:restrict}~\ref{it:restrict:relation}, and
  Remark~\ref{rem:hatR},
  \begin{align*}
    \tilde{s} & \in \ECS[\hat{G}]{\hat\psi} \cap (S/\RR) \\
              & = \ECS[\hat{G}]{[\tilde{s}]_{\hat{\RR}}
                                \anyuntil \hat{T}} \cap (S/\RR) \\
              & = \ECS[G/\RR]{([\tilde{s}]_{\hat{\RR}} \cap (S/\RR))
                              \anyuntil (\hat{T} \cap (S/\RR))} \\
              & = \ECS[G/\RR]{[\tilde{s}]_{\hat\RR|(S/\RR)} \anyuntil T_\RR} \\
              & = \ECS[G/\RR]{\{\tilde{s}\} \anyuntil T_\RR} \ .
  \end{align*}
  Let $\psi \equiv \{\tilde{s}\} \anyuntil T_\RR$ so that $\tilde{s}
  \in \ECS[G/\RR]{\psi}$.
  Then there exists a \dlfree \memoryless\ controller $\overlinit{C}$
  for~$G/\RR$ such that $\langle \overlinit{C}/(G/\RR), \tilde{s} \rangle
  \models \psi$.
  Consider two cases depending on whether $\tilde{s}^\omega \in
  \Frags^\omega(\overlinit{C}/(G/\RR))$ or not.
  \begin{itemize}
  \item
    If $\tilde{s}^\omega \in \Frags^\omega(\overlinit{C}/(G/\RR))$, then 
    $\psi \equiv \{\tilde{s}\} \weakuntil T_\RR$, and
    there exists $\theta \in \Sigma_\RR$ such that $(\tilde{s}, \theta,
    \tilde{s}) \in \delta_\RR$.
    By \eqref{eq:quotient:Sigma} and~\eqref{eq:quotient:delta}, $\theta$ is
    an \RR-superblock step formula of the form $\theta \equiv \tilde{s}
    \weakuntil T$, for some $T \in \SB(\RR)$.
		
  \item
    If $\tilde{s}^\omega \notin \Frags^\omega(\overlinit{C}/(G/\RR))$, then
    since $\overlinit{C}/(G/\RR)$ is \dlfree, there exists $\theta \in
    \ACT_\RR$ and $\tilde{t} \in T_\RR$ such that $(\tilde{s}, \theta,
    \tilde{t}) \in \delta_\RR$ and $\theta \in \overlinit{C}(\tilde{s})$.
    On the other hand, $\tilde{s}^\omega \notin
    \Frags^\omega(\overlinit{C}/(G/\RR))$ implies $(\tilde{s}, \theta,
    \tilde{s}) \notin \delta_\RR$. It follows from
    \eqref{eq:quotient:Sigma} and~\eqref{eq:quotient:delta} that this
    $\theta$ has the form $\theta \equiv \tilde{s} \until T$ for some $T
    \in \SB(\RR)$ with $\tilde{t} \subseteq T$.
  \end{itemize}
  In both cases, there exist $\theta \in \Sigma_\RR$ and $T \in \SB(\RR)$
  such that $\theta \equiv \tilde{s} \anyuntil_\theta T$.
  The modality $\anyuntil_\theta$ of~$\theta$ may be different from the
  modality \anyuntil\ of $\hat\psi$ and~$\psi$, but if $\theta \equiv
  \tilde{s} \weakuntil T$ then also $\psi \equiv \{\tilde{s}\} \weakuntil
  T_\RR$ and $\hat\psi \equiv [\tilde{s}]_{\hat\RR} \weakuntil \hat{T}$.
	
  It is next shown that $T \subseteq \hat{T} \cap S$. Consider an arbitrary
  state $t \in T$. Then $[t]_\RR \subseteq T$ as $T \in \SB(\RR)$, and
  given $\theta \equiv \tilde{s} \anyuntil_\theta T \in \Sigma_\RR$, it
  follows by~\eqref{eq:quotient:delta} that $(\tilde{s}, \theta, [t]_\RR)
  \in \delta_\RR$.
  Also, it follows from $\langle \overlinit{C}/(G/\RR), \tilde{s} \rangle
  \models \psi \equiv \{\tilde{s}\} \anyuntil T_\RR$ that
  $\Post_{\overlinit{C}/(G/\RR)}(\tilde{s}) \subseteq \{\tilde{s}\} \cup
  T_\RR = \{\tilde{s}\} \cup (\hat{T} \cap (S/\RR)) \subseteq \{\tilde{s}\}
  \cup \hat{T}$.
  As $(\tilde{s}, \theta, [t]_\RR) \in \delta_\RR$, this means $[t]_\RR \in
  \{\tilde{s}\} \cup \hat{T}$. Also $\tilde{s} \cap T = \emptyset$ as
  $\theta \in \ACT_\RR$ is an \RR-superblock step formula, which given $t
  \in T$ implies that $t \notin \tilde{s}$, and then $[t]_\RR \neq
  \tilde{s}$ and $[t]_\RR \in \hat{T}$. But $\hat{T} \in \SB(\hat\RR)$ and
  $([t]_\RR,t) \in \hat\RR$ by \defn~\ref{def:hatR}, which means $t \in
  \hat{T}$. As $t \in T$ was chosen arbitrarily, it follows that $T
  \subseteq \hat{T}$.
  Further, note that $T \subseteq S$ as $T \in \SB(\RR)$, and thus also $T
  \subseteq \hat{T} \cap S$.
	
  Now the claim follows as
  \begin{align*}
  s & \in \tilde{s} \\
    & \subseteq \ECS[G]{\tilde{s} \anyuntil_\theta T}
    & \llap{\text{(as $\tilde{s} \anyuntil_\theta T \equiv \theta
                   \in \Sigma_\RR$ and by \eqref{eq:quotient:Sigma})}} \\
    & \subseteq \ECS[G]{[s]_\RR \anyuntil_\theta (\hat{T} \cap S)} 
    & \text{(as $T \subseteq \hat{T} \cap S$)} \\
    & = \ECS[G]{([s]_{\hat\RR} \cap S)
                \anyuntil_\theta (\hat{T} \cap S)} \kern-2em 
    & \text{(by \lemm~\ref{lem:restrict}~\ref{it:restrict:relation})} \\
    & = \ECS[\hat{G}]{[s]_{\hat\RR} \anyuntil_\theta \hat{T}} \cap S
    & \text{(by \lemm~\ref{lem:Gunion})} \\
    & \subseteq \ECS[\hat{G}]{[s]_{\hat\RR} \anyuntil \hat{T}}
    & \kern-2em
      \text{(as $\mathord{\anyuntil}_\theta = \mathord{\weakuntil}$
             implies $\mathord{\anyuntil} = \mathord{\weakuntil}$)} \\
    & = \ECS[\hat{G}]{[\tilde{s}]_{\hat\RR} \anyuntil \hat{T}}
    & \text{(as $(\tilde{s},s) \in \hat\RR$)} \\
    & = \ECS[\hat{G}]{\hat\psi} \ .
  \end{align*}
	
\item $(\tilde{s}, \tilde{t}) \in \hat{\RR}$ where $\tilde{s}, \tilde{t}
  \in \SmodR$. In this case $\tilde{s} = \tilde{t}$ by
  Remark~\ref{rem:hatR}, which implies $L(\tilde{s}) = L(\tilde{t})$ and
  $\tilde{s} \in \ECS[\hat{G}]\psi$ if and
  only if $\tilde{t} \in \ECS[\hat{G}]\psi$ for all~$\psi$.
\end{itemize}

Finally, it is to be shown that $\hat{\RR}$ is a robust stutter
bisimulation between $G$ and $G/\RR$ as defined by
\Defn~\ref{def:robust:stutter:between}. First, if $s\init \in S\init$
then $[s\init]_\RR \in S\init_\RR$ by~\eqref{eq:quotient:Sinit}. By
construction of $\hat{\RR}$, it follows that $(s\init, [s\init]_\RR) \in
\hat{\RR}$. Second, if $\tilde{s}\init \in S\init_\RR$ then there exists
$s\init \in S\init$ such that $\tilde{s} = [s\init]_\RR$, again
by~\eqref{eq:quotient:Sinit}. By construction of $\hat{\RR}$, it follows
that $(\tilde{s}\init, s\init) = ([s\init]_\RR, s\init) \in \hat{\RR}$.
Hence, $\hat{\RR}$ is a robust stutter bisimulation between $G$
and~$\GmodR$.
\end{proof}

\section{Proofs for concrete controller construction}
\label{app:LTLnnPreservation}

This section provides the necessary results to prove \Thm~\ref{thm:LTLnnPreservation}. The definitions of~$T$ and~$\Psi$ from \Sect~\ref{sec:concreteController} are restated here more formally and in more general terms.

\begin{definition}
\label{def:enforcedSuperblock}
Let $G = \tsystem$ be a transition system,
let $\RR \subseteq S \times S$ be an equivalence relation,
let $C$ be a controller for~$G$,
and let $\rho = \rho'u \in \Frags^*(C,G)$ with $\rho' \in S^*$ and $u \in S$.
The \emph{superblock reached} by~$C$ after~$\rho$ is
\begin{equation}\label{eq:TAfter}
  T\langle C, \rho \rangle =
  \bigcup \, \LongSet{10em}{$[v]_\RR \mid
    \rho\tau v \in \Frags^*(C,G)$ where $\tau \in [u]_\RR^*$ and
    $(u, v) \notin \RR$}.
\end{equation}
The \emph{\RR-superblock step formula enforced} by~$C$ after~$\rho$ is
\begin{equation}\label{eq:psiAfter}
  \Psi\langle C, \rho\rangle \equiv 
  \begin{cases}
    [u]_\RR \weakuntil
    T\langle C, \rho \rangle, &  \\
    \multicolumn{2}{r@{}}{\text{if}\ \Frags^\omega(C, G) \cap
                        \rho [u]_\RR^\omega \neq \emptyset\ ;} \\
    [u]_\RR \until T\langle C, \rho \rangle,
    & \quad\text{otherwise .} \\
  \end{cases}
\end{equation}
\end{definition}

The following \Lemm~\ref{lem:psiAfter} shows that the enforced superblock
step formula $\Psi\langle C, \rho\rangle$ is indeed achievable by control,
i.e., there exists a controller enforcing this formula from the equivalence
class at the end of the considered path~$\rho$.

\begin{lemma}
  \label{lem:psiAfter}
  Let $G = \tsystem$ be a transition system,
  let $\RR \subseteq S \times S$ be an equivalence relation,
  and let $C$ be a controller for~$G$,
  and let $\rho = \rho'u \in \Frags^*(C,G)$ with $\rho' \in S^*$ and $u \in S$.
  Then $u \in \EC{\Psi\langle C, \rho\rangle}$.
\end{lemma}

\begin{proof}
\def\Crho{C_\rho}
\def\xtilde{}
\def\xTilde{}
Define a controller $\Crho\colon \xTilde{S}^+ \to 2^\Sigma$
for~$\xTilde{G}$ such that $\Crho(\xtilde{u}\xtilde\tau) =
\xTilde{C}(\xtilde\rho\xtilde\tau)$ for all $\xtilde\tau \in \xTilde{S}^*$.
It will be shown that $\langle
\Crho/\xTilde{G},\xtilde{u}\rangle \models \Psi\langle
\xTilde{C}, \xtilde\rho\rangle$.
Let $\xtilde\pi = \xtilde{u} \xtilde{u}_0 \xtilde{u}_1 \cdots \in
\Frags^\omega(\Crho,\xTilde{G})$, which means
$\xtilde\rho'\xtilde\pi = \xtilde\rho \xtilde{u}_0 \xtilde{u}_1 \cdots
\in\Frags^\omega(\xTilde{C},\xTilde{G})$.
Consider two cases: either all the states $\xtilde{u}_i$ are equivalent
to~$\xtilde{u}$ or not.
	
If $(\xtilde{u},\xtilde{u}_i) \in \RR$ for all $i$, then
$\xtilde\rho'\xtilde\pi \in \Frags^\omega(\xTilde{C},\xTilde{G})$ with
$\xtilde\pi \in [\xtilde{u}]_\RR^\omega$, and then $\xtilde\pi \models
[\xtilde{u}]_\RR \weakuntil T\langle \xTilde{C}, \xtilde\rho\rangle \equiv
\Psi\langle \xTilde{C}, \psi\rangle$ by~\eqref{eq:psiAfter}.

Otherwise $\xtilde\pi$ can be written as $\xtilde\pi = \xtilde{u} \xtilde{u}_0
\cdots \xtilde{u}_m \xtilde{v} \xtilde\pi'$ where $(\xtilde{u},\xtilde{u}_i) \in
\RR$ for $0 \leq i \leq m$ and $(\xtilde{u},\xtilde{v}) \notin \RR$. Then
$\xtilde\rho \xtilde{u}_0 \cdots \xtilde{u}_m \xtilde{v} \in
\Frags^*(\xTilde{C},\xTilde{G})$ and $\xtilde{v} \in T\langle \xTilde{C},
\xtilde\rho\rangle$ by~\eqref{eq:TAfter}. As $\xtilde{u} \xtilde{u}_0 \cdots
\xtilde{u}_m \in [\xtilde{u}]_\RR^*$, it is clear that $\xtilde\pi = \xtilde{u}
\xtilde{u}_0 \cdots \xtilde{u}_m \xtilde{v} \xtilde\pi' \models [\xtilde{u}]_\RR
\anyuntil T\langle \xTilde{C}, \xtilde\rho\rangle \equiv \Psi\langle
\xTilde{C}, \xtilde\rho\rangle$ independently of which case
of~\eqref{eq:psiAfter} applies.
	
It has been shown that $\xtilde\pi \models \Psi\langle \xTilde{C},
\xtilde\rho\rangle$ for every infinite path fragment $\xtilde\pi \in
\Frags^\omega(\Crho,\xTilde{G})$ whose first state
is~$\xtilde{u}$.
Therefore $\langle \Crho/\xTilde{G},\xtilde{u}\rangle
\models \Psi\langle \xTilde{C}, \xtilde\rho\rangle$, which implies the claim
$\xtilde{u} \in \EC[\smash{\xTilde{G}}]{\Psi\langle \xTilde{C},
\xtilde\rho\rangle}$.
\end{proof}

During the equivalence proof of \thm~\ref{thm:LTLnnPreservation}, it needs
to be established that the paths visited by the controller~$\Tilde{C}$ and
the controller~$C$ from \defn~\ref{def:concrete:controller} are not only
stutter equivalent but also visit the same equivalence classes.
Given a transition system~$G$ with state set~$S$, an equivalence relation $\RR \subseteq S \times S$, and a sequence of states $\pi = s_0s_1\cdots \in S^\infty$, the \emph{trace of equivalence classes} of $\pi$ is $[\pi]_\RR = [s_0]_\RR [s_1]_\RR \cdots \in (S/\RR)^\infty$. Then two path fragments $\pi_1, \pi_2 \in \Frags^\infty(G)$ are \emph{\RR-stutter equivalent} if $\mathrm{sf}([\pi_1]_\RR) = \mathrm{sf}([\pi_2]_\RR)$. In other words, two path fragments are \RR-stutter equivalent if they both visit the same equivalence classes in the same order, albeit possibly taking different numbers of steps within equivalence classes. If \RR\ is a \rdsb, then by \defn~\ref{def:robust:stutter:bisimulation}~\ref{it:rsb:label} all infinite \RR-stutter equivalent path fragments are also stutter equivalent.

Furthermore, the mapping~$M$ from \eqref{eq:M:init}--\eqref{eq:M:step} is
extended to infinite path fragments using a closure construction.
For a sequence $\pi \in S^\infty$, the set of finite prefixes of~$\pi$ is
$\pfx(\pi) = \{\, \rho \in S^+ \mid \rho \sqsubseteq \pi \,\}$.
The \emph{closure} of a set of finite sequences $P \subseteq S^*$ is
the set of sequences whose prefixes are all in~$P$, defined by $\clo(P) =
\{\, \pi \in S^\infty \mid \pfx(\pi) \subseteq P \,\}$.
Then the extended map $M\colon S^\infty \to S^\infty$ is
\begin{equation}
  \label{eq:M:extended}
  M(\pi) = \clo(\{ M(\rho) \mid s \sqsubseteq \rho  \sqsubset \pi \}) \ .
\end{equation}
This extension is well-defined because~$M$ is a prefix-preserving map.
If $M(\pi) = \tilde{\pi}$ is infinite, then it is an infinite path fragment
in $\Frags^\omega(\Tilde{C}, G)$ with $\tilde{s} \sqsubseteq \tilde{\pi}$,
and $\pi$ and~$\tilde{\pi}$ are \RR-stutter equivalent.

Using these definitions, it is now possible to prove
\Thm~\ref{thm:LTLnnPreservation} and its \crl~\ref{cor:LTLnn:preservation}.

\begin{repeattheorem}{thm:LTLnnPreservation}
	\ThmLTLnnPreservation
\end{repeattheorem}

\begin{proof}
Let $(\tilde{s}, s) \in \mathcal{R}$ such that $\tilde{s} \in
\EC[G]{\varphi}$. Then there exits a controller $\Tilde{C}$ such that $\langle
\Tilde{C}/G, \tilde{s} \rangle \models \varphi$. To show $s \in
\EC[G]{\varphi}$, it is shown that the controller~$C$ as constructed in
\defn~\ref{def:concrete:controller} enforces~$\varphi$ from~$s$, i.e.,
$\langle C/G, s \rangle \models \varphi$.
	
The proof consists of two parts. First it is shown that the mapping~$M$ in
\eqref{eq:M:init}--\eqref{eq:M:step} is defined for every finite path
fragment permitted by~$C$. Then it is shown that every infinite path
fragment permitted by~$C$ starting in~$s$ is \RR-stutter equivalent to some
infinite path fragment permitted by~$\Tilde{C}$ that starts in~$\tilde{s}$.
	
For the first step, it is shown by induction over the length of path
fragments $\rho \in \Frags^*(C, G)$ with $s \sqsubseteq \rho$ that
\begin{enumerate}
\item\label{it:LTLnnPreservation:M:def}
  $M(\rho)$ is defined;
\item\label{it:LTLnnPreservation:M:frag}
  $M(\rho) \in \Frags^*(\Tilde{C},G)$;
\item\label{it:LTLnnPreservation:M:eqv}
  $M(\rho)$ and $\rho$ are \RR-stutter equivalent.
\end{enumerate}

In the base case, the only path fragment~$\rho$ of length 1 such that $s
\sqsubseteq \rho$ is $\rho = s$, in which case $M(\rho) = M(s) = \tilde{s}$
is defined by construction~\eqref{eq:M:init}. Also, it holds that $\tilde{s} \in
\Frags^*(\Tilde{C},G)$ for the path fragment of length~1, and $\tilde{s}$
and~$s$ are \RR-stutter equivalent as $(\tilde{s}, s) \in \mathcal{R}$
	
For the inductive step, consider $\rho = \rho' u v \in \Frags^*(C, G)$ with
$\rho' \in S^*$, $u, v \in S$, and $s \sqsubseteq \rho$. By inductive
assumption, $M(\rho'u)$ satisfies
\ref{it:LTLnnPreservation:M:def}--\ref{it:LTLnnPreservation:M:eqv}, and
because~$M$ is a map to~$S^+$, this can be written as $M(\rho'u) =
\tilde\rho'\tilde{u} \in \Frags^*(\Tilde{C},G)$. There are two cases to
consider depending on whether $(u, v) \in \RR$ or $(u, v) \notin \RR$.
	
If $(u, v) \in \RR$, it follows by construction~\eqref{eq:M:stay} and by
inductive assumption that $M(\rho) = M(\rho'u v) = M(\rho'u) =
\tilde\rho'\tilde{u} \in \Frags^*(\Tilde{C},G)$ is defined.
Given $(u, v) \in \RR$, it is also clear the $\rho = \rho' u v$ is
\RR-stutter equivalent to $\rho'u$, which is \RR-stutter equivalent to 
$M(\rho) = M(\rho'u)$ by inductive assumption.
	
In the second case, $(u, v) \notin \RR$, it must be shown that
$F(M(\rho'u),v) = F(\tilde\rho'\tilde{u}, v)$ in~\eqref{eq:M:step} is
nonempty. By \Defn~\ref{def:concrete:controller}
\begin{equation}
  \label{eq:LTLnnPreservation:barC}
  C(\rho'u) = \overlinit{C}\langle \psi \rangle (u) \ ,
\end{equation}
where $\psi \equiv \Psi\langle \Tilde{C}, M(\rho'u) \rangle \equiv \Psi
\langle \Tilde{C}, \tilde\rho' \tilde{u} \rangle$.
As $\tilde\rho'\tilde{u} \in \Frags^*(\Tilde{C},G)$, it follows by
\Lemm~\ref{lem:psiAfter} and \Propn~\ref{prop:EC:any} that $\tilde{u}
\in \EC{\psi} = \ECS{\psi}$. Since $\rho = \rho'u$ and $M(\rho) =
\tilde\rho'\tilde{u}$ are \RR-stutter equivalent, it follows that
$(\tilde{u}, u) \in \RR$, and thus $u \in \ECS{\psi}$ as \RR\ is a robust
stutter bisimulation.
Recall that $v \notin [u]_\RR$ as $(u,v) \notin \RR$ and $\rho' u v \in
\Frags^*(C,G)$, which implies $u v \in
\Frags^*(\overlinit{C}\langle\psi\rangle,G)$
by~\eqref{eq:LTLnnPreservation:barC}, and then $v \in T\langle \Tilde{C},
\tilde\rho'\tilde{u} \rangle$ as $\overlinit{C}\langle\psi\rangle$ enforces
$\psi \equiv \Psi\langle \Tilde{C}, \tilde\rho' \tilde{u} \rangle \equiv
[\tilde{u}]_\RR \anyuntil T \langle \Tilde{C}, \tilde\rho'\tilde{u}
\rangle$ for paths starting from~$u$.
Then by~\eqref{eq:TAfter} there exist $\tilde{\tau} \in [\tilde{u}]_\RR^*$
and $(\tilde{v}, v) \in \RR$ such that $\tilde\rho'\tilde{u} \tilde{\tau}
\tilde{v} \in \Frags^*(\Tilde{C}, G)$.
Then $\tilde\rho'\tilde{u} \tilde{\tau} \tilde{v} \in
F(\tilde\rho'\tilde{u}, v)$ by~\eqref{eq:F}, so $F(\tilde\rho'\tilde{u},
v)$ is nonempty, which is enough to show~\ref{it:LTLnnPreservation:M:def}.
Also $M(\rho' u v) \in F(\tilde\rho'\tilde{u}, v) \subseteq
\Frags^*(\Tilde{C},G)$ by \eqref{eq:F} and~\eqref{eq:M:step},
showing~\ref{it:LTLnnPreservation:M:frag}.
Lastly, any $M(\rho' u v) \in F(\tilde\rho'\tilde{u}, v)$ is \RR-stutter
equivalent to $\tilde\rho'\tilde{u}\tilde{v}$ by~\eqref{eq:F},
so \ref{it:LTLnnPreservation:M:eqv} follows as $(\tilde{v}, v) \in \RR$, and
$\tilde\rho'\tilde{u} = M(\rho'u)$ is \RR-stutter equivalent to~$\rho'u$ by
inductive assumption.
This completes the induction.
	
Finally it can be shown that $\langle C/G, s \rangle \vDash \varphi$. This
amounts to showing that any infinite path fragment $\pi \in
\Frags^\omega(C, G)$ with $s \sqsubseteq \pi$ satisfies~$\varphi$. This is
done by considering two cases, depending on whether $\pi$ is mapped to an
infinite or finite path fragment by the extended map~$M$
from~\eqref{eq:M:extended}.

If $M(\pi) = \tilde\pi$ is infinite, then $\rho \in \Frags^*(C,G)$ for
every finite prefix $\rho \sqsubset \pi$ with $s \prefix \rho$, and thus
$M(\rho) \in \Frags^*(\Tilde{C},G)$ by~\ref{it:LTLnnPreservation:M:frag},
which implies $M(\pi) = \clo (\{\, M(\rho) \mid s \prefix \rho \sqsubset
\pi \,\}) \in \Frags^\omega(\Tilde{C},G)$ by~\eqref{eq:M:extended}.
Also, every finite prefix $\rho \sqsubset \pi$ with $s \prefix \rho$ is
\RR-stutter equivalent to~$M(\rho)$ by~\ref{it:LTLnnPreservation:M:eqv},
which again using~\eqref{eq:M:extended} implies that $\tilde\pi = M(\pi)$
is \RR-stutter equivalent to~$\pi$.
	
If $M(\pi) = \tilde{\rho}$ is finite, then $\tilde{\rho} \in
\Frags^*(\Tilde{C}, G)$ with $\tilde{s} \sqsubseteq \tilde{\rho}$, and
there is a finite prefix $\rho \sqsubset \pi$ such that $\tilde{\rho} =
M(\pi) = M(\rho)$. Let $\rho = \rho'u$ and $\tilde{\rho} =
\tilde\rho'\tilde{u}$.
As $\rho$ and $M(\rho) = \tilde{\rho}$ are \RR-stutter equivalent
by~\ref{it:LTLnnPreservation:M:eqv}, it follows that $(u, \tilde{u}) \in
\RR$.
By construction of~$M$ in \eqref{eq:M:stay}, \eqref{eq:M:step},
and~\eqref{eq:M:extended}, it follows from the finiteness of $M(\pi) =
M(\rho'u)$ that there exists $\tau \in [u]_\RR^\omega$ such that $\pi =
\rho' u \tau$. Then also $M(\rho'\tau'u') = M(\rho'u) = M(\rho)$ for all
prefixes $\tau'u' \prefix u\tau$.
For such prefixes it holds that
$\overlinit{C}\langle\Psi\langle\Tilde{C},M(\rho'\tau'u')\rangle\rangle(u')
= \overlinit{C}\langle\Psi\langle\Tilde{C},M(\rho)\rangle(u') =
\overlinit{C}\langle\psi\rangle(u')$ where $\psi \equiv
\Psi\langle\Tilde{C}, M(\rho)\rangle$.
This implies $u\tau \in \Frags^\omega(\overlinit{C}\langle \psi \rangle, G)$.
As $\langle\overlinit{C}\langle\psi\rangle/G, u\rangle \models \psi$
by \defn~\ref{def:concrete:emulator}, it follows that $u\tau \models \psi$.
Furthermore, $\psi \equiv \Psi\langle\Tilde{C}, M(\rho)\rangle \equiv
\Psi\langle\Tilde{C}, M(\rho'u) \rangle$ is of the form $\psi \equiv
[u]_\RR \anyuntil T\langle\Tilde{C},\rho\rangle$ by~\eqref{eq:psiAfter}.
Given $u\tau \in [u]_\RR^\omega$, it follows from $u\tau \models \psi$ by
the semantics of \LTLnn that $\psi \equiv [u]_\RR \weakuntil
T\langle\Tilde{C},\rho\rangle$.
Also note $\psi \equiv \Psi\langle\Tilde{C}, M(\rho)\rangle \equiv
\Psi\langle\Tilde{C}, \tilde\rho\rangle \equiv \Psi\langle\Tilde{C},
\tilde\rho'\tilde{u}\rangle$, so it follows by
construction~\eqref{eq:psiAfter} that $\Frags^\omega(\Tilde{C},G) \cap
\tilde\rho[\tilde{u}]_\RR^\omega \neq \emptyset$.
Thus there exists $\tilde{\tau} \in [\tilde{u}]_\RR^\omega$ such that
$\tilde\rho\tilde\tau \in \Frags^\omega(\Tilde{C}, G)$.
Let $\tilde\pi = \tilde\rho\tilde\tau$. Then $\tilde\pi =
\tilde\rho\tilde\tau \in
\Frags^\omega(\Tilde{C}, G)$, and $\tilde\pi = \tilde\rho\tilde\tau =
\tilde\rho'\tilde{u}\tilde\tau$ is \RR-stutter equivalent to
$\tilde\rho'\tilde{u} = \tilde\rho = M(\pi)$, which is \RR-stutter
equivalent to~$\pi$ by \ref{it:LTLnnPreservation:M:eqv}
and~\eqref{eq:M:extended}.
	
In both cases where an arbitrary infinite path fragment~$\pi$ is picked
from $\Frags^\omega(C, G)$, there exists an \RR-stutter equivalent path
fragment~$\tilde{\pi} \in \Frags^\omega(\Tilde{C}, G)$. As \RR\ is a robust
stutter bisimulation, it preserves the state labels, and it follows that
$\pi$ and~$\tilde{\pi}$ are stutter equivalent. Since $\tilde{\pi} \vDash
\varphi$ and $\varphi$ is an \LTLnn\ formula, it also holds that $\pi
\vDash \varphi$. Because~$\pi$ with $s \prefix \pi$ is picked arbitrarily,
$\langle C/G, s \rangle \vDash \varphi$, and hence $s \in \EC{\varphi}$.
\end{proof}

\begin{repeatcorollary}{cor:LTLnn:preservation}
  \CorLTLnnPreservation
\end{repeatcorollary}

\begin{proof}
Assume $G = \tsystem$ and $\Tilde{G} = \langle \Tilde{S}\bcom
\Tilde\Sigma\bcom \Tilde\delta\bcom \Tilde{S}\init\bcom \Tilde\AP\bcom
\Tilde{L}\rangle$.
Consider an arbitrary state $s \in S\init$.
Since \RR\ is a robust stutter bisimulation between $G$ and~$\Tilde{G}$, it
is a robust stutter bisimulation on $G' = G \cup \Tilde{G}$. Therefore, for
$s \in S\init$ there exists $\tilde{s} \in \Tilde{S}\init$ such that $(s,
\tilde{s}) \in \RR$. As $\Tilde{C}/\Tilde{G} \models \varphi$, it holds
that $\tilde{s} \in
\EC[\smash{\Tilde{G}}]{\varphi} \subseteq \EC[G']{\varphi}$. Then also $s
\in \EC[G']{\varphi}$ as $(s, \tilde{s}) \in \RR$, so there exists a \dlfree
controller $C_s \colon (S \cup \Tilde{S})^+ \to 2^\Sigma$ such that
$\langle C_s/G', s\rangle \models \varphi$.

Define the controller $C \colon S^+ \to 2^\Sigma$ such that $C(s\rho') =
C_s(\rho')$ if $s \in S\init$ and $C(s\rho') = \Sigma$ if $s \in S
\setminus S\init$. That is, when $C$ is faced with a path~$\rho = s\rho'$
starting at some state $s \in S\init$, it emulates the control action of
the controller~$C_s$ that enforces~$\varphi$ starting from this state~$s$.
It is clear that $C$ is \dlfree, and $\langle C/G, s\rangle \models
\varphi$ holds for all states $s \in S\init$, which means $C/G \models
\varphi$.
\end{proof}

As another consequence of \thm~\ref{thm:LTLnnPreservation}, the following
\lemm~\ref{lem:superblock} shows that robust stutter bisimulation preserves
the existence of controllers for general superblock step formulas. That is,
condition~\ref{it:rsb:equivalent} in
\defn~\ref{def:robust:stutter:bisimulation} does not only apply to
superblock step formulas $[s]_\RR \anyuntil T$ from an equivalence class
$[s]_\RR \in S/\RR$, but to general superblock step formulas $P \anyuntil
T$ with $P \in \SB(\RR)$.

\begin{lemma}[Superblock Lemma]
\label{lem:superblock}
Let $G = \tsystem$ be a transition system, let $\RR$ be a robust stutter bisimulation on $G$, and let $\psi$ be an $\RR$-superblock step formula for~$G$. For all states $(s, \tilde{s}) \in \RR$ such that $s \in \ECS[G]{\psi}$ it also holds that $\tilde{s} \in \ECS[G]{\psi}$.
\end{lemma}

\begin{proof}
Let $\psi \equiv P \anyuntil T$ be an \RR-superblock step formula.
Construct a transition system $G' = \langle S\bcom \Sigma\bcom \delta\bcom
S\init\bcom \AP'\bcom L'\rangle$ where $\AP' = \{ p, t \}$ and $L'$ is
defined such that
$p \in L'(u)$ iff $u \in P$ and $t \in L'(u)$ iff $u \in T$. In this
transition system, $\psi \equiv P \anyuntil T$ is logically equivalent to
the \LTLnn\ formula $p \anyuntil t$. It is also clear that \RR\ is a robust stutter bisimulation on~$G'$: the states and transitions in $G$ and~$G'$ are the same, and because $P$ and~$T$ are superblocks of~\RR\ (and hence closed under equivalence), any equivalent states are assigned the same labels by~$L'$.

Now let $(s, \tilde{s}) \in \RR$ such that $s \in \ECS[G]{\psi}$. As $\psi$ does not refer to any proposition in~$\AP$, it follows that $\ECS[G]{\psi} = \ECS[G']{\psi}$. Furthermore, $\ECS[G']{\psi} = \EC[G']{\psi}$ by \Propn~\ref{prop:EC:any}. Hence, $s \in \EC[G']{\psi} = \EC[G']{p \anyuntil t}$ because of logical equivalence in $G'$. As $(s, \tilde{s}) \in \RR$ and $\RR$ is a robust stutter bisimulation on $G'$, it follows by \Thm~\ref{thm:LTLnnPreservation} that $\tilde{s} \in \EC[G']{p \anyuntil t} = \EC[G']{\psi} = \ECS[G']{\psi} = \ECS[G]{\psi}$, again by \Propn~\ref{prop:EC:any} and because $\psi$ does not refer to any atomic propositions in \AP.
\end{proof}

\section{Coarsest robust stutter bisimulation}
\label{app:coarsest}

This section proves that the relation~\inbisim\ from
\defn~\ref{def:inbisim} is the coarsest robust stutter bisimulation on a
transition system. They key argument in this proof is that the transitive
and reflexive closure of the union of robust stutter bisimulations is again
a robust stutter bisimulation.

The \emph{transitive closure}~$(\RR)^+$ of a relation~\RR\ is the smallest
transitive relation such that $\RR \subseteq (\RR)^+$.
The \emph{transitive and reflexive closure} of~\RR\ is $(\RR)^* = (\RR)^+
\cup \mathcal{I}$, where $\mathcal{I}$ is the identity relation.

\begin{lemma}
	\label{lem:union}
	Let $G$ be a transition system, and let $\{\, \RR_i \mid i \in I \,\}$ be a set of robust stutter bisimulations on~$G$, where $I$ is an arbitrary index set. Then $(\bigcup_{i \in I} \RR_i)^*$ is a robust stutter bisimulation on~$G$.
\end{lemma}

\begin{proof}
As the union of symmetric relations is again symmetric, it is clear that
the reflexive and transitive closure of the union of equivalence
relations~$\RR_i$ is again an equivalence relation.
It remains to be shown that $\RR = (\bigcup_{i \in I} \RR_i)^*$ satisfies
conditions \ref{it:rsb:label} and~\ref{it:rsb:equivalent} in
\defn~\ref{def:robust:stutter:bisimulation}.
	
	Let $(x,y) \in \RR$. This can be written as
	\begin{equation}
		x = x_0 \inRR_{i_1} x_1 \inRR_{i_2} \cdots \inRR_{i_n} x_n = y
	\end{equation}
	where $i_1,\ldots,i_n \in I$.
	As each $\RR_i$ is a robust stutter bisimulation, it immediately follows that
	$L(x) = L(x_0) = L(x_1) = \cdots = L(x_n) = L(y)$, showing~\ref{it:rsb:label}.
	To show condition~\ref{it:rsb:equivalent},
	let $x = x_0 \in \ECS[G]\psi$ for some \RR-superblock step formula $\psi \equiv [x_0]_\RR \anyuntil T$.
	If $n = 0$, it immediately holds that $y = x_n = x_0 \in \ECS[G]\psi$.
	Otherwise notice that $\RR_{i_1} \subseteq \bigcup_{i \in I} \RR_i \subseteq (\bigcup_{i \in I} \RR_i)^* = \RR$.
	Then $[x_0]_\RR$ and~$T$ are superblocks of~$\RR_{i_1}$, i.e., $\psi$ is an $\RR_{i_1}$-superblock step formula.
	Since $\RR_{i_1}$ is a robust stutter bisimulation and $(x_0,x_1) \in \RR_{i_1}$, it follows by \Lemm~\ref{lem:superblock} that $x_1 \in \ECS[G]\psi$.
	By induction, it follows that $y = x_n \in \ECS[G]\psi$.
\end{proof}

An alternative way to define the coarsest \rdsb~\inbisim\ is by considering
the union of all \rdsb{s}. This union is again a \rdsb, so it must be the
coarsest such relation. This argument is made in the following proof.

\begin{proposition}
\label{prop:inbisim}
Let $G$ be a transition system. The relation~\inbisim\ is a \rdsb that
includes every \rdsb on~$G$.
\end{proposition}

\begin{proof}
By \defn~\ref{def:inbisim}, the relation~\inbisim\ can be written as
\begin{equation}
  \inbisim = \bigcup \; \LongSet{9em}{\RR\ $\mid$ \RR\ is a robust stutter
    bisimulation on~$G$}.
  \label{eq:Runion}
\end{equation}
Thus, \inbisim\ includes every \rdsb on~$G$. To prove that \inbisim\ is a
\rdsb, by \lemm~\ref{lem:union} it is enough to show $\inbisim =
(\inbisim)^*$. It is clear that $\inbisim \subseteq (\inbisim)^*$ by
definition of closure. Conversely, as $(\inbisim)^*$ is a \rdsb, it follows
from~\eqref{eq:Runion} that $(\inbisim)^* \subseteq \inbisim$.
\end{proof}

\section{Correctness of algorithm}
\label{app:algorithm}

This section contains the proofs of the propositions in
Section~\ref{sec:algorithm}, which show the correctness of
Algorithm~\ref{alg:QuotientPartition}. \Propn~\ref{prop:splitter} shows
that splitters only exists for relations that are not \rdsb{s}, and
\propn~\ref{prop:refine:coarseness} shows that refinement preserves the
property of a relation to be at least as coarse as the coarsest
\rdsb~\inbisim. \Propn~\ref{prop:algorithm} combines these results to
show that Algorithm~\ref{alg:QuotientPartition} produces the coarsest
\rdsb\ if it terminates.

\begin{repeatproposition}{prop:splitter}
  \PropSplitter
\end{repeatproposition}

\begin{proof}
Assume that \RR\ is a robust stutter bisimulation on~$G$.
Then~\RR\ satisfies condition~\ref{it:rsb:label} in
\defn~\ref{def:robust:stutter:bisimulation}. If there was
a splitter~$\psi$ of~\RR, then by \Defn~\ref{def:splitter} there would exist
$(s_1, s_2) \in \RR$ such that $s_1 \in \ECS{\psi}$ and $s_2 \notin
\ECS{\psi}$. However, this contradicts condition~\ref{it:rsb:equivalent} of
\defn~\ref{def:robust:stutter:bisimulation}.
	
Conversely, assume that \RR\ satisfies condition~\ref{it:rsb:label} in
\defn~\ref{def:robust:stutter:bisimulation} and there does not exist any
splitter of~\RR. It is enough to show the condition~\ref{it:rsb:equivalent}
of \defn~\ref{def:robust:stutter:bisimulation}. Therefore, let $(s_1, s_2)
\in \RR$ and $s_1 \in \ECS{\psi}$ for some \RR-superblock step formula
$\psi \equiv [s_1]_\RR \anyuntil T$. Clearly $\ECS{\psi} \cap [s_1]_\RR
\neq \emptyset$. As it was assumed that there are no splitters, $\psi$
cannot be a splitter, so by \defn~\ref{def:splitter} it follows that
$\ECS{\psi} \cap [s_1]_\RR = [s_1]_\RR$. As $(s_1, s_2) \in \RR$, it
follows that $s_2 \in [s_1]_\RR = \ECS{\psi} \cap [s_1]_\RR \subseteq
\ECS{\psi}$.
\end{proof}

\begin{repeatproposition}{prop:refine:coarseness}
  \PropRefineCoarseness
\end{repeatproposition}

\begin{proof}
Assume $\inbisim \subseteq \RR$, and consider two robust stutter bisimilar
states $s_1 \bisim s_2$. It is to be shown that $(s_1, s_2) \in
\Refine(\RR, \psi)$. As $\psi \equiv P \anyuntil T$ is an \RR-superblock
step formula, $P$ and $T$ are superblocks of \RR, and then also of the
finer relation \inbisim. This means that~$\psi$ is an \inbisim-superblock
step formula for~$G$. Note that $s_1 \bisim s_2$ implies $(s_1, s_2) \in
\RR$ by assumption $\inbisim \subseteq \RR$, so $[s_1]_\RR = [s_2]_\RR$.
Consider two cases.
	
If $[s_1]_\RR \neq P$ then note that $P \in S/\RR$ by
\defn~\ref{def:splitter}, so that $[s_1]_\RR \cap P = \emptyset$. Then
$s_1,s_2 \notin P \supseteq P \cap \ECS\psi$ and thus $(s_1, s_2) \notin D
\cup D^{-1}$ in~\eqref{eq:refine}, so that $(s_1, s_2) \in \RR \setminus (D
\cup D^{-1}) = \Refine(\RR, \psi)$.
	
Otherwise $[s_1]_\RR = P$, and there are two further cases: either $s_1 \in
\ECS{\psi}$ or $s_1 \notin \ECS{\psi}$. If $s_1 \in \ECS{\psi}$, then since
$s_1 \bisim s_2$ and \inbisim\ is a robust stutter bisimulation on~$G$
and~$\psi$ is a \inbisim-superblock step formula for~$G$, it follows from
\Lemm~\ref{lem:superblock} that $s_2 \in \ECS{\psi}$. Thus $s_1,s_2 \notin
P \setminus \ECS\psi$ and $(s_1, s_2) \notin D \cup D^{-1}$
in~\eqref{eq:refine}, so $(s_1, s_2) \in \Refine(\RR, \psi)$. Similarly, if
$s_1 \notin \ECS{\psi}$ it follows that $s_2 \notin \ECS{\psi}$ so that
$s_1,s_2 \notin P \cap \ECS\psi$ and again $(s_1, s_2) \in \Refine(\RR,
\psi)$.
\end{proof}

\begin{repeatproposition}{prop:algorithm}
  \PropAlgorithm
\end{repeatproposition}

\begin{proof}
It is clear that $\RR^0$ satisfies condition~\ref{it:rsb:label} in
\defn~\ref{def:robust:stutter:bisimulation}, and each iteration of the loop
produces a finer equivalence relation that continues to satisfy this
condition. If the loop terminates, then the loop entry condition on
line~\ref{l:loop} of Algorithm~\ref{alg:QuotientPartition} ensures that
there are no splitters of the result~$\RR^i$ and it follows from
\Propn~\ref{prop:splitter} that $\RR^i$ is a robust stutter bisimulation.
Furthermore, the initial relation $\RR^0$ is the coarsest equivalence
relation that satisfies condition~\ref{it:rsb:label} in
\defn~\ref{def:robust:stutter:bisimulation} and thus coarser than~\inbisim,
and then it follows by induction from \Propn~\ref{prop:refine:coarseness}
that $\RR^i$ is also coarser than~\inbisim\ for all $i \geq 0$. On
termination, $\RR^i \supseteq \inbisim$ is a robust stutter bisimulation,
and as \inbisim\ is the coarsest robust stutter bisimulation, it follows
that $\RR^i = \inbisim$.
\end{proof}
\fi

\end{document}